%% file: SMO.tex
\documentclass{article}
\usepackage{graphicx}
\usepackage{amsmath,amssymb,amsfonts,amsthm}
\usepackage{mathrsfs}
\usepackage[hang,small,bf]{caption}
\usepackage{subfigure}
\usepackage[shortlabels]{enumitem}
\usepackage{bm}
\usepackage[a4paper,top=25mm,bottom=25mm,inner=30mm,outer=20mm]{geometry}
\usepackage{tikz}
\usetikzlibrary{arrows,calc}
\usetikzlibrary{arrows,shapes}
\usetikzlibrary{matrix}
\tikzset{
>=stealth',
help lines/.style={dashed, thick},
axis/.style={<->},
important line/.style={thick},
connection/.style={thick, dotted},
}
\usepackage{hyperref}
\input{macros}

\begin{document}

\title{Optimal design of a micro-tubular fuel cell}
\author{
G. Delgado\textsuperscript{1}
\footnote{Current address: IRT SystemX, Palaiseau, France (gabriel.delgado@irt-systemx.fr)}
.
}
\maketitle
\begin{center}
\emph{
\textsuperscript{1} 
CMAP, \'Ecole Polytechnique, Palaiseau, France.
}

\end{center}

\begin{abstract}
We discuss the problem of the optimal design of a micro-tubular fuel cell applying an inverse homogenization technique. Fuel cells are extremely clean and efficient electrochemical power generation devices, made up of a cathode/electrolyte/anode structure, whose energetic potential has not being fully exploited in propulsion systems in aeronautics due to their low power densities. Nevertheless, thanks to the recent development of additive layer manufacturing techniques (3D printing), 
complex structures usually impossible to design with conventional manufacturing techniques 
can be constructed with a low cost, allowing notably to build porous or foam-type structures for fuel cells. We seek thus to come up with the micro-structure of an arrangement of micro-tubular cathodes which maximizes the contact surface {subject} to a pressure drop and a permeability constraint. The optimal periodic design (fluid/solid) emerges from the application of a shape gradient algorithm coupled to a level-set method for the geometrical description of the corresponding cell problem.\\

\end{abstract}


\section{Introduction}

Fuel cells are energy conversion devices which can continuously convert chemical energy into electrical energy and heat, without involving direct combustion. This feature offers many advantages over traditional power sources such as improved efficiency, greater fuel diversity, high scalability, no moving parts (hence less noise and vibration) and lower impact on the environment \cite{gou2009fuel}.
 
A fuel cell is a fairly simple device, mainly composed of three elements: an anode, a cathode, and an electrolyte between the two electrodes. The two electrodes are connected together by an electrical circuit. On the surface of these electrodes, electrochemical reactants react through half-redox reactions, producing (or consuming) ions, electrons, and in most cases, heat. Ions pass through the electrolyte meanwhile  electrons are ``channeled" in the electric circuit and then routed to the second electrode to be consumed.  

The efficiency of current fuel cells, which ranges from $40\%$ to $60\%$, is higher than thermal systems such as gas turbines  
since their operation is not constrained by any theoretical thermodynamic limitation as the maximum Carnot efficiency. Furthermore, coupled to a gas turbine at a high temperature ($800-1000^o$C), the spawned hybrid system can achieve really high efficiencies (near $85\%$), producing electricity from the waste of heat. This feature makes high temperature operating fuel cells an interesting complement to aircraft engines \cite{palsson2000combined,singhal2000advances,samuelsen2004fuel,roth2010fuel,steffen2005solid}.
\begin{figure}[h]
\centering
\includegraphics[trim=0 2cm 0 5cm,clip,width=0.9\textwidth]{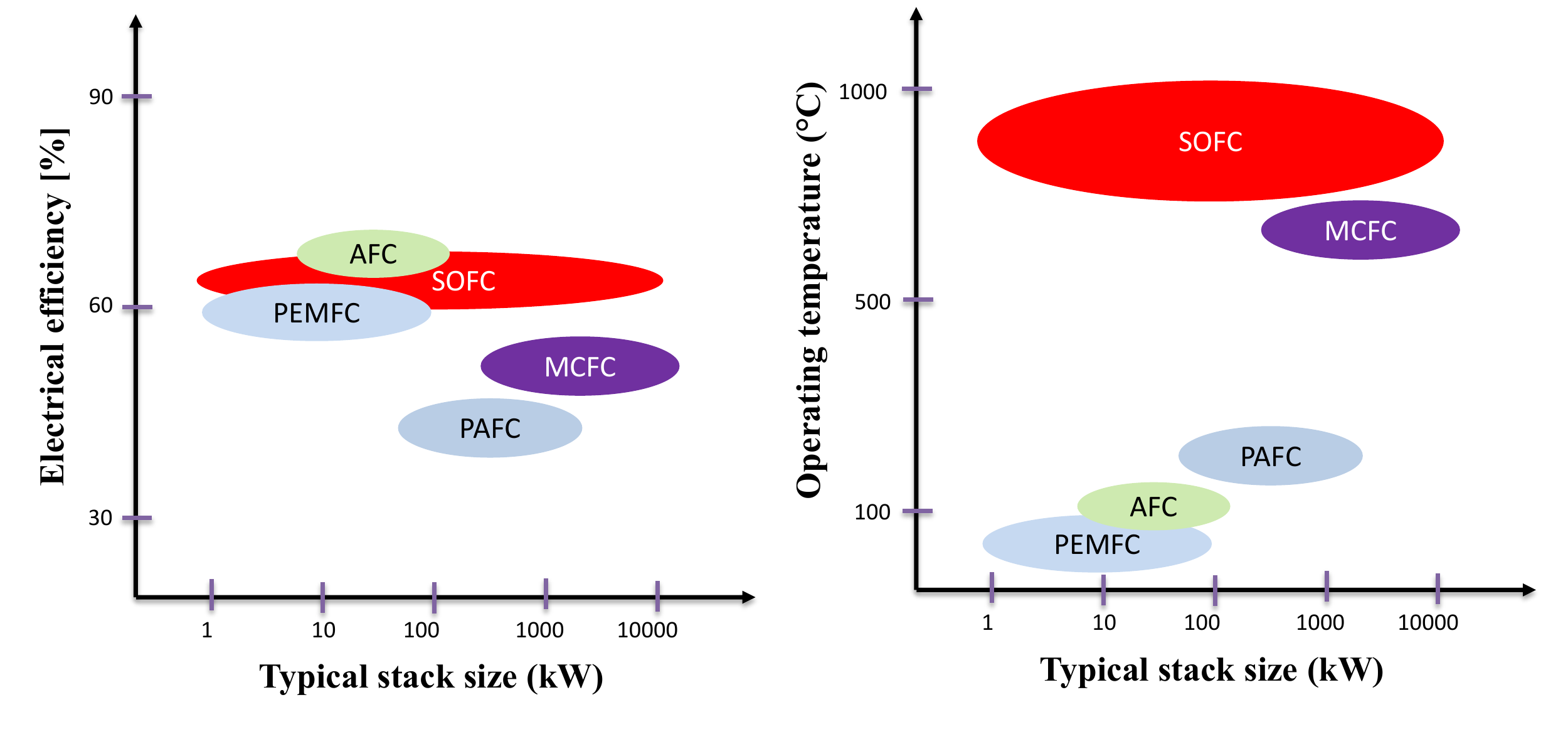}
\caption{Different types of fuel cells: Polymeric Electrolyte Membrane (PEMFC), Alkaline (AFC), Phosphoric Acid (PAFC), Molten Carbonate (MCFC) and Solid Oxide (SOFC). Reprinted with permission from \cite{etienne}.}\label{FC_performances}
\end{figure}

Among the different types of fuel cells (see Fig. \ref{FC_performances}), one can recognize the Solid Oxide Fuel Cells (SOFC) as a particularly appealing model for high temperature applications in hybrid systems.
SOFC possesses various advantages w.r.t other fuel cells such as a solid electrolyte allowing different geometries and shapes, good performance and durability, and high operating temperature ($800-1000^o C$) for reforming. This last attribute has the potential of hybridization with a gas turbine, {as it is shown} in Fig. \ref{hybrid}. 

\begin{figure}[h]
\centering
\includegraphics[scale=0.5]{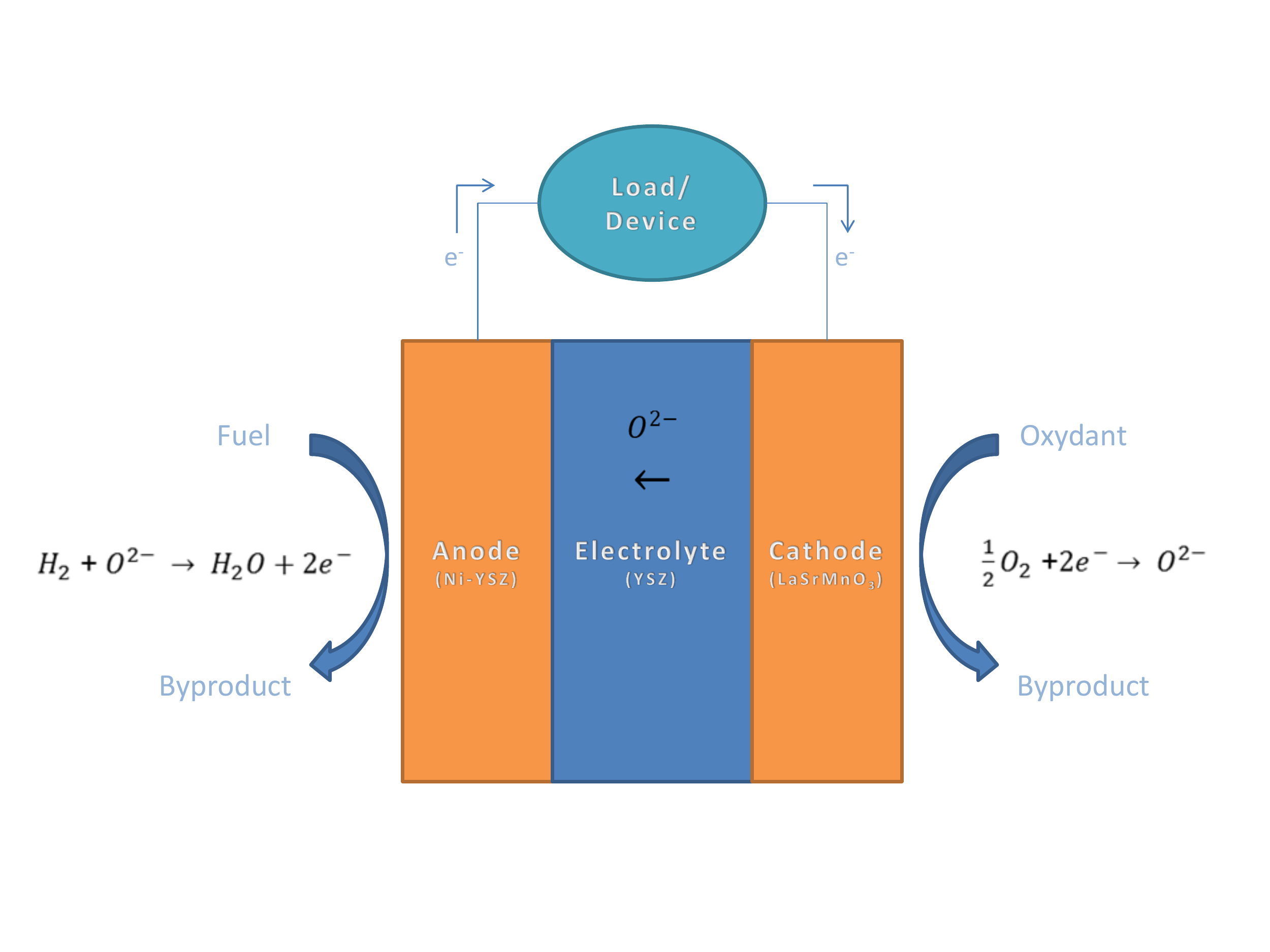}
\caption{Triple structure of a solid oxide fuel cell. The $O_2$ from the air is consumed in the cathode, meanwhile the $O^{2-}$ ions liberated from the reduction reaction travel through the electrolyte to the anode, where the oxidation reaction of the fuel ($H_2$) takes place, liberating heat, water and electricity. Reprinted with permission from \cite{etienne}.}\label{sofc}
\end{figure}

\begin{figure}[hbt]
\centering
\includegraphics[trim=0 0.5cm 0 4cm,clip,width=0.8\textwidth]{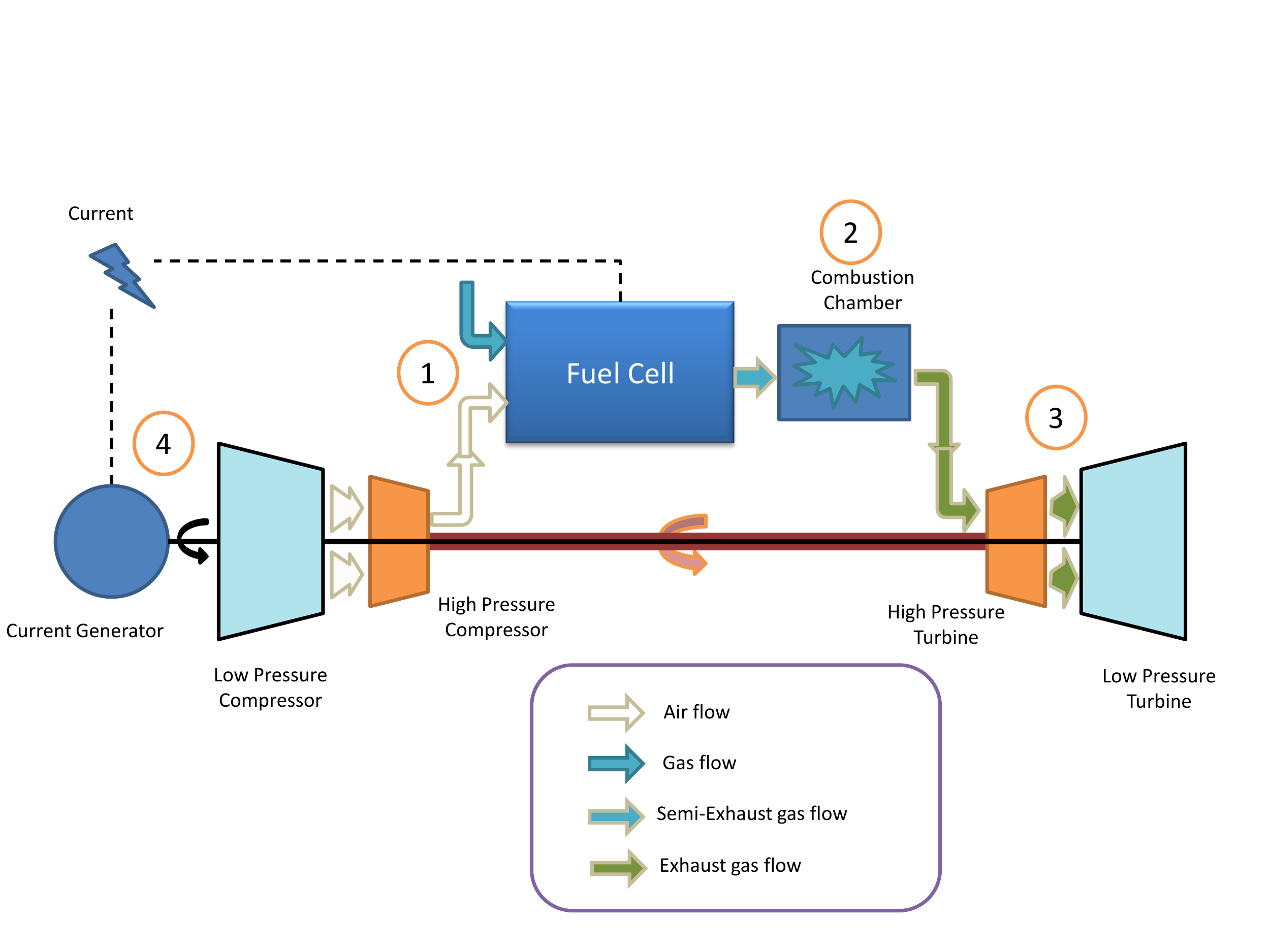}
\caption{Example of an hybrid system. Reprinted with permission from \cite{etienne}.}\label{hybrid}
\end{figure}

The cathode of an SOFC is usually an alloy of lanthanum, strontium, and manganese oxide. On its surface takes place the half reduction reaction of oxygen, producing an oxygen ion. The electrolyte is e.g. made of YSZ (yttrium stabilized zirconia). This allows the ion transport to the anode (nickel mixed with YSZ), in contact with which it reacts with hydrogen to produce water and electrons. Finally, these electrons are conducted and used by an electric device (Figure \ref{sofc}). Both the anode and the cathode must be very porous so as to allow the transport of the fuel and the oxygen, respectively. 
%

SOFC can be designed following many geometric configurations. However, the most common designs are the planar and the tubular ones. Each design offers advantages and drawbacks depending of the application requirement. Planar design configuration has a low physical component volume profile and short current path between single cells, allowing higher power densities, meanwhile tubular designs have e.g. high thermo-mechanical properties, simple sealing requirements and good thermal shock resistance \cite{de2011production}.  

Despite of the impressive energy efficiencies achieved by the SOFC, the foregoing geometries lack of an essential property required by any aircraft system: high gravimetric and volumetric energy densities. Indeed, fuel cells are still too heavy to propel any large aircraft since they have a lower power density when compared with conventional turbines \cite{hordeski2009hydrogen}. For instance, the gravimetric power density (kW/kg) of a SOFC compared to the turbojet CFM-56 (one of the most popular turbojets in the world and often used as a reference for studies of this type) is at least five times smaller.

In order to improve the fuel cell efficiency, various alternatives have emerged, such as miniaturization of the structure through a micro-tubular configuration \cite{de2011production} and the use of porous materials (foam) in the design of some components (e.g. bipolar/end plates) \cite{kumar2003modeling}. In general these techniques take advantage 
of recent developments in additive layer manufacturing or 3D printing procedures, which allow the construction of extremely complicated 3D structures from a CAD model.

\begin{figure}[h]
\centering
\includegraphics[trim=0 1cm 0 3cm,clip,scale=0.4]{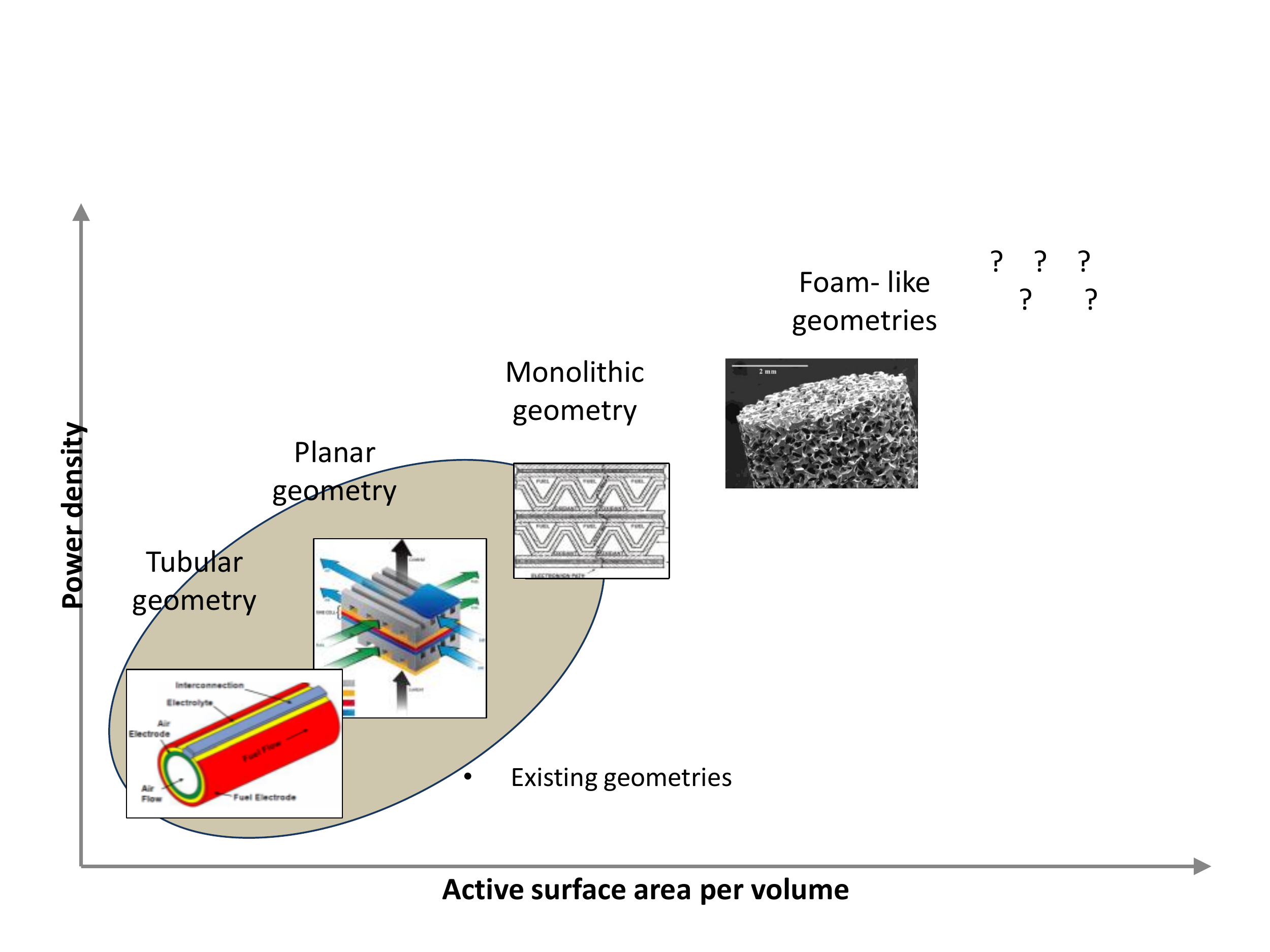}
\caption{Evolution of power density vs active surface of SOFC. Reprinted with permission from \cite{etienne}.}\label{evolution}
\end{figure}

Thus, the design of fuel cells with very small features, leading notably to periodic patterns, becomes a foremost challenge in the construction of future technologies (see Figure \ref{evolution}). An adapted tool for this purpose is inverse homogenization. The homogenization method for topology optimization, which has been successfully implemented in structural optimal design (see e.g. \cite{allaire2002shape,bendsoe1995methods,cherkaev2000variational,allaire1993optimal,bendsoe1988generating}), consists in admitting composites regions where the void appears at a micro-scale. The theoretical foundations are given in \cite{kohn1986optimal,tartar2009general}. 
Meanwhile in homogenization the effective properties of a material are found from the micro-structure, in inverse homogenization the micro-structure does not exist initially but we seek to come up with a micro-structure with prescribed or extreme homogenized properties. The design of materials with extreme or prescribed properties using inverse homogenization in elasticity, fluid mechanics, wave propagation, etc., is well known in the literature \cite{cadman2013design,sigmund1994materials,sigmund1994design,haslinger1995optimum,sigmund1997design,zhou2008computational,guest2007design,bendsoe2004topology}.

In the process of finding the optimal micro-structure, a topology optimization problem must be solved within the so-called cell problem. The cell problem defines the link between the physical macroscopic (effective) properties and the micro-structure characterized by a certain geometry. Among the most popular methods in topology optimization for fluid flow problems, the level-set method for shape optimization \cite{allaire2004structural,burger2003framework,sethian2000structural,wang2003level} arises as a viable, robust and efficient alternative to more standard density-based approaches \cite{allaire2002shape,allaire1997shape,bendsoe1995methods}. First introduced in \cite{osher.sethian}, the level-set method has the advantage of tracking the interfaces on a fixed mesh, easily managing topological changes without any need of re-meshing. Allied to the Hadamard method of shape differentiation, the level-set approach gives a better description and control of the geometrical properties of the interface without need of any intermediate density, avoiding typical drawbacks such as intermediate density penalization and possible spurious physical behavior during the optimization process.

In the present article we propose an optimal configuration of a micro-tubular fuel cell, whose cathode constitutes a tubular periodic structure designed by inverse homogenization. This particular structure leads to a maximal surface micro-structure subjected to a permeability constraint for the air phase and a pressure drop constraint for the fuel phase. In Section \ref{fuel_cell_setting_problem}, the main features of the problem are described, namely the physical modeling of a porous fuel cell, the homogenization of the equations governing the system and the optimization problem we contemplate to solve. Then the shape gradients of the functionals involved in the optimization problem are recalled in Section \ref{shape_deriv_fuel_cell}. Finally an example of an optimal periodical micro-structure (fluid/solid) is detailed in Section \ref{numerical_results_fuel_cell}, emerging from the application of the level-set method for topology optimization to the corresponding cell problem.

\section{Problem setting}
\label{fuel_cell_setting_problem}
\subsection{Physical modeling of a porous fuel cell}

A generic SOFC, regardless of the geometrical configuration (planar, tubular, monolithic, etc), is always composed of two porous electrodes (anode and cathode), a dense electrolyte, an anodic and cathodic gas channel and two current collectors. However, for the sake of simplicity, we propose to study a reduced model of a micro-tubular SOFC \cite{de2011production}, as described in {Figure \ref{tubular_fuel_cell}}, focusing on the three former components. Thus, we consider a periodically perforated domain, 
where the air flows freely around a periodic collection of cathode/electrolyte/anode tubes. The air cannot penetrate inside the tubes and the fuel flows orthogonally through them.
The air transports several species but we concentrate on the $O_2$, which is consumed in the electrochemical reaction on the surface of the cathode.

\begin{figure}[h]
\begin{center}
\fbox{\includegraphics[trim=0 1cm 0 3cm,clip,scale=0.4]{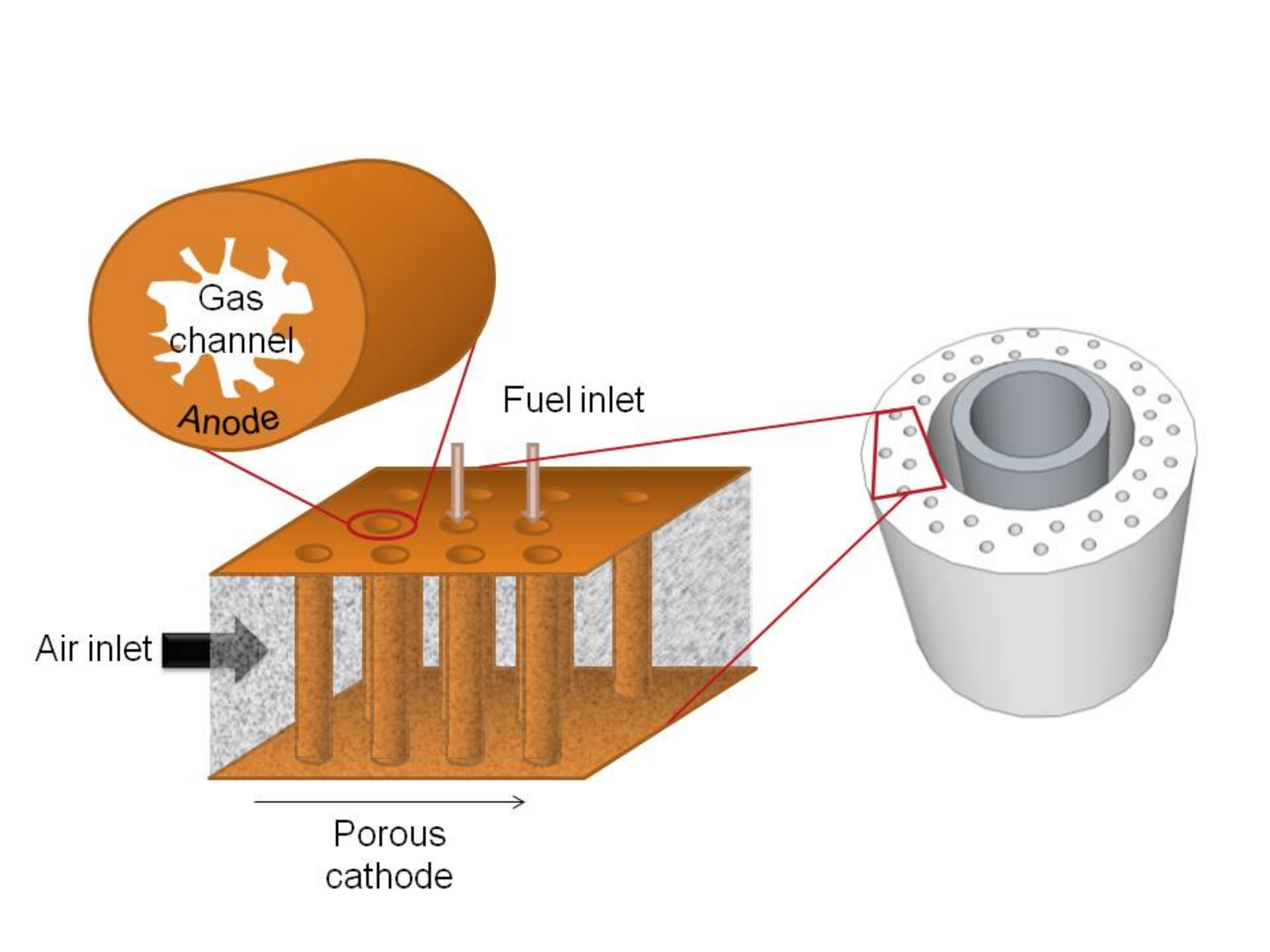}}
\end{center}
\caption{Micro-tubular design with arbitrary shaped tubes proposed in \cite{etienne} (reprinted with permission). See also the similar micro-tubular SOFC design of AIST \cite{national_institute}.}\label{tubular_fuel_cell}
\end{figure}

 \paragraph{Assumptions.} Now we list the main physical assumptions of our model, similarly to \cite{bove2008modeling}{.}
\begin{enumerate}[(1)]
\item \textit{Steady state velocity of the air (also called slow steady state)}.

\item \textit{Laminar and incompressible flow.} This assumption derives from the low gas speed in the SOFC gas channels, where the density variation of each specie is not related to compressions/expansions, but rather caused either internally by heat release of chemical reactions or externally by wall heating and by mass variation. Furthermore, considering that the gas speed in SOFC gas channels is always very low (low Mach number $<0.3$), it is a common practice to assume a laminar flow in the gas channels, meaning that the non-linear term in the momentum conservation equation (inertia term) is negligible with respect to the viscous term.

\item \textit{Isothermal state.} Since the temperature of the SOFC is quite high ($800^o-1000^o$C), the local variations can be underestimated. This allows in particular to avoid the dependence of the diffusion tensor and the reaction ratio of $O_2$ with respect to the temperature.

\item \textit{The electrochemical reactions are confined to the electrode-electrolyte interface.} The place where the electrochemical reaction takes place, the so-called triple-phase-boundary, is the site where ions, electrons and gas coexist, thus enabling the redox reactions. Since this place represents a small portion of the entire electrode domain and the high electronic conductivity of the electrodes compared to the ionic conductivity, the redox reactions are likely to take place very close to the electrode-electrolyte interface.

\item \textit{The cathode/electrolyte/anode is a lumped structure so it can be treated as one interface.} 

\end{enumerate}

Generally speaking, the physic of the SOFC is modeled as a complex system where chemical reactions, electrical conduction, ionic conduction, gas phase mass transport, and heat transfer take place simultaneously and are tightly coupled \cite{hosseini2013cfd,hussain2006mathematical,danilov2009cfd}. Nevertheless, we simplify the physical description by focusing only on the gas mass transport and the chemical reaction on the cathode.
\begin{figure}[h]
\centering
\includegraphics[width=0.4\textwidth]{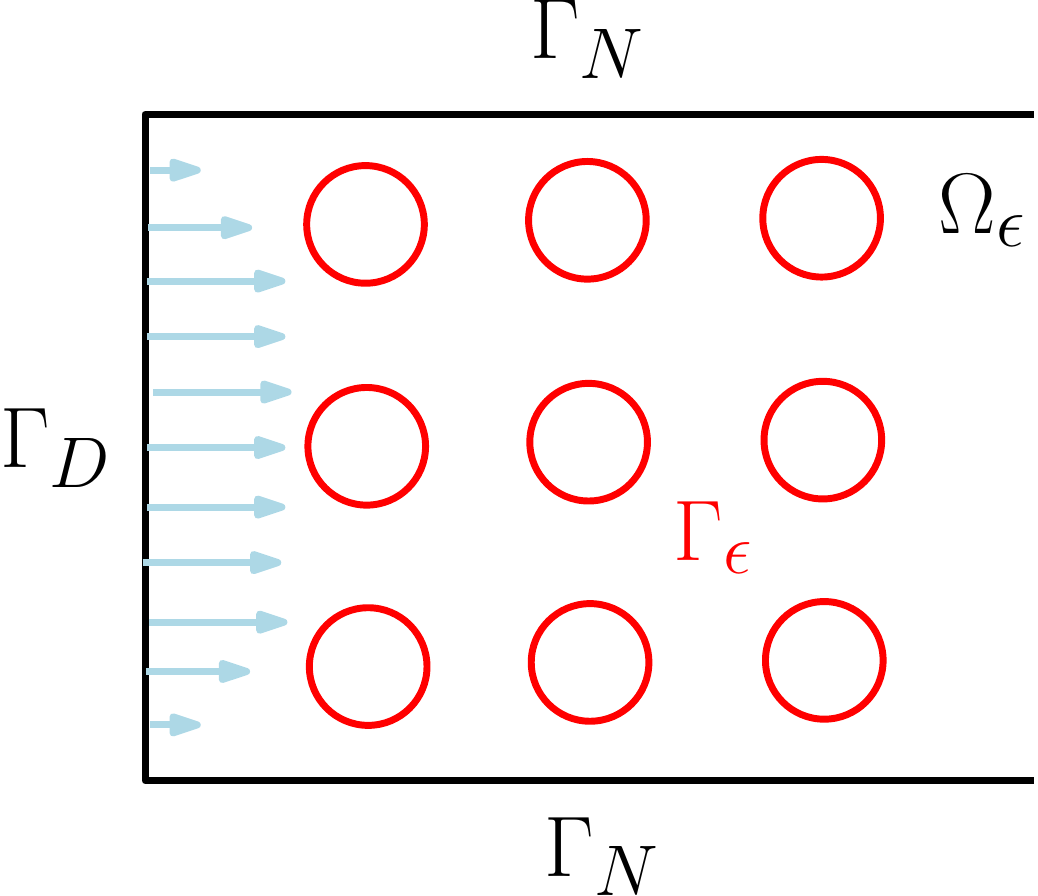}
\caption{Transversal section of the porous SOFC and boundary conditions. The air flows from the left to the right. The fuel penetrates orthogonally through the figure inside the red tubes. $\Gamma_\epsilon$ represents the anode/electrolyte/cathode lumped structure.}\label{porous_medium}
\end{figure}

{Consider Fig. \ref{porous_medium}.} Let $\Omega_\epsilon \subset \Rbb^d$ be the porous volume of the fuel cell where the air flows freely, delimited by the boundaries $\Gamma_D$, $\Gamma_N$ and $\Gamma_{\epsilon}$. The air is injected through $\Gamma_D$, meanwhile $\Gamma_N$ is impermeable. A periodic arrangement of anode/electrolyte/cathode tubes constitutes the interface $\Gamma_\epsilon$. The redox reaction takes place on $\Gamma_\epsilon$ and the adimensional parameter $\epsilon$ corresponds to the ratio between the characteristic size of each tube and a macroscopic characteristic length. 

\paragraph{Conservation and constitutive laws.}
 The fluid adimensionalized equations (mass conservation and momentum conservation) plus the specie ($O_2$) diffusion-convection-(surface)reaction adimensionalized equation respectively read \cite{auriault2010homogenization}
 %

\begin{equation}
\label{equation_eps}
\left\{
\begin{array}{ll}
\text{div}(\ue)=0  &  x\in \Omega_\epsilon,\\
-\epsilon^2 \mu \Delta \ue = \nabla \pe & x\in \Omega_\epsilon,\\
\frac{\partial X_\epsilon}{\partial t}+ \ue\cdot \nabla X_\epsilon=\lambda \Delta X_\epsilon 
 & t>0, x\in \Omega_\epsilon,\\
 X_\epsilon =X_{\text{init}} & t=0, x\in \Omega_\epsilon,
\end{array}
\right.
\end{equation}
where $X_\epsilon$, $X_{\text{init}}$ represent the current and the initial concentrations of $O_2$ in the air, $u_\epsilon$ is the local velocity of the air, $p_\epsilon$ the local pressure, $\mu$ the viscosity of the air and
$\lambda$ the diffusion coefficient. The parameter $\epsilon$ as it was explained above, corresponds to the ratio between the characteristic size of each tube and a macroscopic characteristic length. We remark that in this model, the diffusion and the convection of the $O_2$ in the  transport equation are equilibrated at the macro-scale.\\

%

%
\paragraph{Boundary conditions}

\begin{equation}
\label{boundary_conditions_eps}
\left\{
\begin{array}{lll}
\ue=u_D & X_\epsilon=X_D & t>0,x\in \Gamma_D,\\
\ue\cdot n=0 & D_\epsilon \nabla X_\epsilon \cdot n=0 & t>0, x\in \Gamma_N,\\
\ue=0 &  \nabla X_\epsilon \cdot n=-\epsilon \frac{1}{4e}\mathcal{R}(X_\epsilon) & t>0, x\in \Gamma_\epsilon,\\
\end{array}
\right.
\end{equation}
where $(u_D,X_D)$ are the velocity of the air and the concentration of $O_2$ at the inflow boundary $\Gamma_D$,
$e$ is the electron charge and $\mathcal{R}(X_\epsilon)$ is the Butler-Volmer reaction term \cite{Kee2012331}, which can be simplified  \cite{schmuck2013homogenization,bove2008modeling} to
\begin{equation*}
\mathcal{R}(X_\epsilon)=j_0 (X^m_\epsilon,T) \exp\Big(-\alpha_c \frac{zF}{RT}\eta\Big).
\end{equation*}
{Here,} $j_0$ corresponds to the exchange current density and depends linearly upon $X_{\epsilon}^m$, $m$ is the order of reaction,  $\eta$ is the over-potential with respect the Nerst equilibrium potential, $\alpha_c$ is the cathodic transfert coefficient, $T$ the temperature, $z$ the number of electrons involved in the electrode reaction, $R$ the universal gas constant and $F$ the Faraday constant. All the above parameters are assumed to be constants.

 We remark that the coefficient $m$ is positive and in order to prove the existence of a weak solution of \eqref{equation_eps}, \eqref{boundary_conditions_eps}, $m$ should be smaller than a certain critical value $\bar{m}$, which depends on the dimension $d$ of the space. By taking e.g. $m\leq 1$, $\mathcal{R}(X_\epsilon)\in L^2(\Gamma_\epsilon)$ and the variational formulation of \eqref{equation_eps}, \eqref{boundary_conditions_eps} will be well posed. Typical values for $m$ are e.g. $0.4$ in \cite{kenney2004mathematical}, for a composite cathode made of $60\%$ volume of LSM (lanthanum-strontium-manganite) and $40\%$ YSZ, and $0.25$ in \cite{costamagna2004electrochemical,bove2008modeling}. Thus, depending on the physical characteristics of the cathode, the reaction term $\mathcal{R}$ is in general non-linear.


%

%

\subsection{Homogenized system}

Before stating the homogenization of the system \eqref{equation_eps},\eqref{boundary_conditions_eps}, let us describe more precisely the assumptions on the porous domain $\Omega_\epsilon$ \cite{allaire2002shape}. As usual in periodic homogenization, a periodic structure is defined by a domain $\Omega$ and an associated micro-structure, or periodic cell $Y=(0,1)^d$, which is made of two complementary parts : the solid part $\omega$ and the fluid part $Y\backslash \omega$ (see Figure \ref{unit_cell}). 
We assume that $Y\backslash\omega$ is a smooth and connected open subset of $Y$, identified with the unit torus (i.e. $Y\backslash\omega$, repeated by $Y-$periodicity in $\Rbb^d$, is a smooth and connected open set of $\Rbb^d$). The domain $\Omega$ is covered by a regular mesh of size $\epsilon$: each cell $Y_i^{\epsilon}$ is of the type $(0,\epsilon)^d$, and is divided in solid part $\omega_i^{\epsilon}$ and fluid part $Y^{\epsilon}_i\backslash \omega_i^{\epsilon}$, i.e. is similar to the unit cell Y rescaled to size $\epsilon$. The fluid part 
$\Omega_{\epsilon}$ of a porous medium is defined by
\begin{equation*}
\Omega_{\epsilon}=\Omega\backslash \bigcup_{i=1}^{N(\epsilon)} \omega_i^{\epsilon},
\end{equation*}
where the number of cells is $N(\epsilon)=|\Omega|\epsilon^{-d}(1+o(1))$.\\

\begin{figure}[h]
\centering
\includegraphics[width=0.4\textwidth]{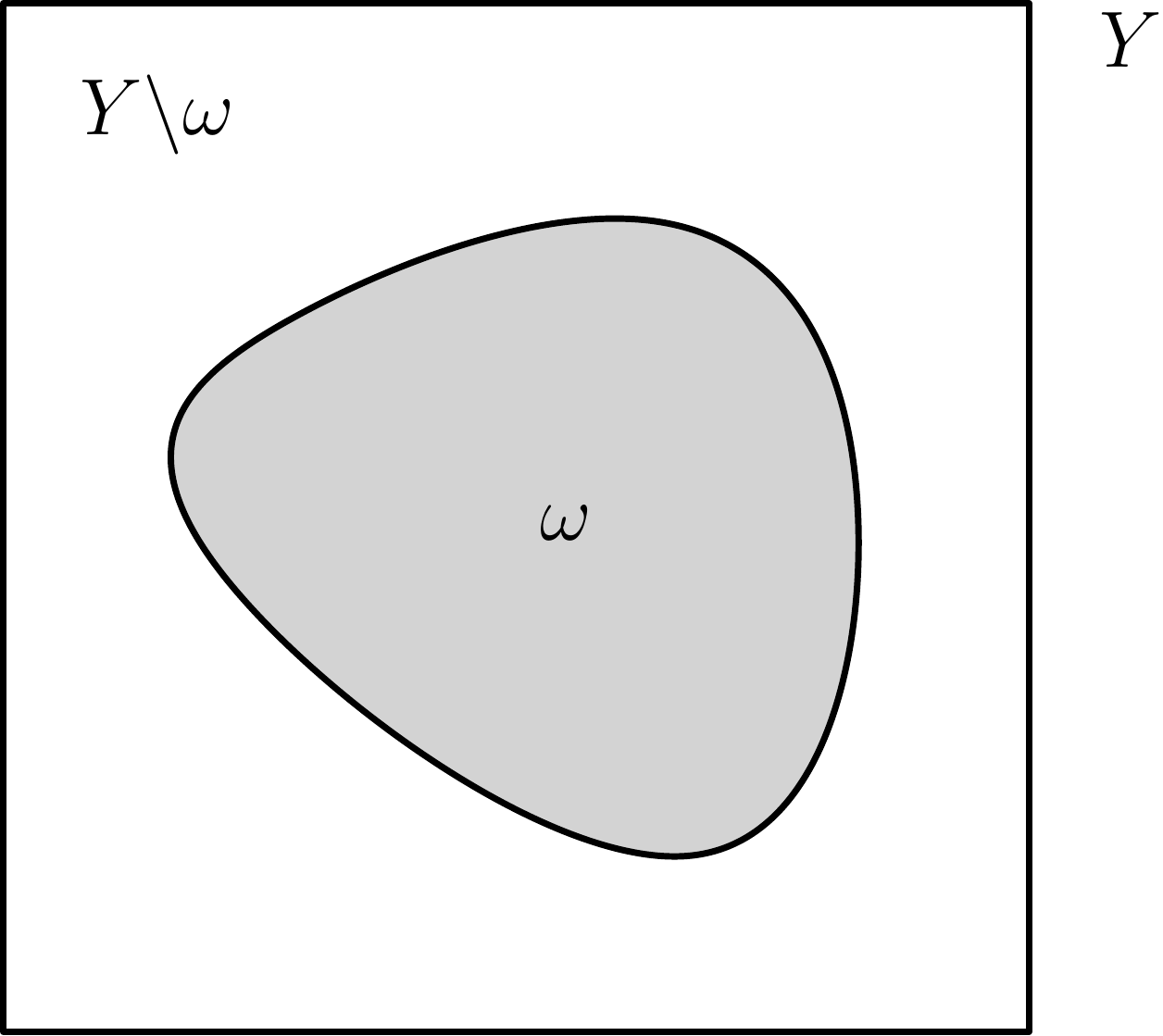}
\caption{Unit cell of a porous medium.}\label{unit_cell}
\end{figure}

Now according to \cite{hornung1997homogenization,hornung1991diffusion} for general heterogeneous catalysis and \cite{schmuck2013homogenization} for PEM fuel cells, the homogenization of the initial boundary value problem \eqref{equation_eps},\eqref{boundary_conditions_eps} corresponds to the following system 
\begin{equation}
\label{equation_homo}
\left\{
\begin{array}{ll}
\mbox{div}(u^*)=0 & x\in \Omega\\
u^*= -\frac{K}{\mu} \nabla p^* & x\in \Omega\\
\frac{\partial X^*}{\partial t}+ u^*\cdot \nabla  X^*=\lambda \mbox{div} (D \nabla X^*) +|\partial \omega| \mathcal{R}(X^*) & x\in \Omega\\
X^* =X_{\text{init}} & t=0, x\in \Omega,
\end{array}
\right.
\end{equation}

and the boundary conditions

\begin{equation}
\label{boundary_conditions_homo}
\left\{
\begin{array}{lll}
u^*=u_D & X^*=X_D & t>0,x\in \Gamma_D,\\
u^*\cdot n=0 & \nabla X^* \cdot n=0 & t>0, x\in \Gamma_N.
\end{array}
\right.
\end{equation}

The functions $u^*,p^*,X^*$ correspond to the homogenized velocity, pressure and $O_2$ concentration, respectively. 
$D$ is the effective porous media diffusion tensor, $|\partial \omega|$ is the perimeter of the micro-fuel-tube $\omega$ scaled to the unit cell and $K$ is the permeability tensor.\\ 


The definition of the tensors $K$ and $D$ stems from the so-called cell problems \cite{hornung1997homogenization}, namely 

\begin{equation}
\label{def:tensors}
K_{ij}=\int_{Y\backslash \omega}\nabla u_i(y) : \nabla u_j(y)dy, \quad D_{ik}=
\int_{Y\backslash \omega}\Big(e_i+\nabla \pi_i\Big)\cdot \Big(e_k+\nabla \pi_k\Big) dy
\end{equation}

where the tensorial product ``$\;:$'' represents the twice-contracted tensorial product, i.e. for $A,B\in \Rbb^{d\times d}$
\[
A:B=\sum \limits_{k,\ell}A_{k\ell}B_{k\ell},
\]
and the respective cell problems read
\begin{equation}
\label{cell_problem}
\left\{ \begin{array}{rcll}
&&\nabla p_i-\Delta u_i=e_i, \mbox{ }x\in Y\backslash \omega,\\
&&\mbox{div}(u_i)=0, \mbox{ }x\in Y\backslash \omega,\\
&&u_i=0,\mbox{ } x\in \partial \omega,\\
&&y\rightarrow p_i(y),u_i(y) \mbox{ Y-periodic,}
\end{array} \right.
\hspace{1cm}
\left\{ \begin{array}{rcll}
&&-\div(\nabla \pi_j+e_j)=0 , \mbox{ }x\in Y\backslash \omega,\\
&&\nabla \pi_j \cdot n=-e_j \cdot n,\mbox{ } x\in \partial \omega,\\
&&y\rightarrow \pi_j(y) \mbox{ Y-periodic,}
\end{array} \right.
\end{equation}

with $(e_i)_{1\leq i\leq d}$ being the canonical basis of $\mathbb{R}^d$.\\

\subsection{Optimization problem}
  
The topology optimization problem {consists} in finding the optimal lay-out $\omega^*\subset Y$, representing the scaled shape of the fuel tubes inside the unit-cell, which solves
  
\begin{equation}
\label{optim_problem}
\left\{ \begin{array}{rcll}
&&\max\limits_{\omega\subset Y} |\partial \omega|\\
&&\mbox{s.t.}\\
 &&|\omega| \geq C_f|\partial \omega|\,\\ 
&& \frac{tr(K)}{d}\geq k_{min},\\ 
\end{array} \right.
\end{equation}

where the tensor $K$ is defined according to \eqref{def:tensors} and $tr()$ is the trace operator. The first constraint represents a lower bound $C_f\geq 0$ of the fuel tube hydraulic diameter, which prevents a drastic pressure drop and a ``too oscillating'' boundary. The constant $C_f$ must satisfy $C_f\leq \sqrt{1/4\pi}$, according to the  isoperimetric inequality. The second constraint involves a lower bound $k_{min}$ of the trace of $K$ (which is intended to give a {measure} of the permeability \cite{guest2007design}) avoiding a high pressure drop in the cathode. In other words, we try to find the shape $\omega$ with the largest perimeter (so it maximizes the factor on the electro-chemical reaction term $\mathcal{R}$ of the homogenized transport equation in \eqref{equation_homo}) such that two constraints of pressure drop inside and outside the tubes are fulfilled.\\

Let $\omega_0\subset Y$ be a regular fixed open set. Then the following result yields the existence of an optimal solution of \eqref{optim_problem}
\begin{prop}
Denote as $\Uadmis$ the collection of open subsets of $Y$ close to $\omega_0$ in the sense of the pseudo-distance
\begin{equation*}
d^{\mathbb{D}}(\omega_0,\omega)= \inf \limits_{T\in \mathbb{D}| T(\omega_0)=\omega}\Big(\left\|T-Id\right\|_{W^{2,\infty}}+ \left\|T^{-1}-Id\right\|_{W^{2,\infty}}\Big),
\end{equation*}
where $\mathbb{D}$ is the set of diffeomorphism 
\begin{equation*}
\mathbb{D}= \left\{T \mbox{ such that }(T-Id)\in W^{2,\infty}(\Rbb^d;\Rbb^d),(T^{-1}-Id)\in W^{2,\infty}(\Rbb^d;\Rbb^d)  \right\}.
\end{equation*}
Furthermore, chose in particular the family of applications $T=Id+\theta$, where $\theta \in W^{2,\infty}(\Rbb^d;\Rbb^d)$ and $\left\|\theta \right\|_{W^{2,\infty}}<1$. Then problem \eqref{optim_problem} admits at least one optimal solution in $\Uadmis$.

\end{prop}
\begin{proof}
It is a classical result that the set $\Uadmis$ is compact w.r.t. the above topology, introduced by F.Murat and J.Simon \cite{murat1974quelques}. Furthermore since $\left\|\theta \right\|_{W^{2,\infty}}<1$, any minimizing sequence of \eqref{optim_problem} has a convergent subsequence in $W^{1,\infty}(\Rbb^d;\Rbb^d)$, due to the Sobolev embedding theorem. Finally the existence of a solution of \eqref{optim_problem} stems from the continuity of the perimeter, the area and the permeability tensor $K$ w.r.t. $\left\|\theta\right\|_{W^{1,\infty}}$, thanks to their respective shape differentiability. The concept of the shape derivative will be recalled in the next section.\\ 
\end{proof}

\section{Shape sensitivity analysis}
\label{shape_deriv_fuel_cell}

In this section we briefly recall the definition of the shape derivative and  
the gradients of the functionals involved in (\ref{optim_problem}) (the proofs can be found in \cite{schmidt2010shape}).

Shape differentiation is a classical topic that goes back to Hadamard \cite{allaire2006conception,pierre,intro_sokolowski}. 
Let the overall domain $\Omega\subset \mathbb{R}^d$ be fixed and bounded. 
Let $\omega \subset \Omega$ be a smooth open subset which is variable. Indeed, 
we consider variations of the type
$$
\big( Id + \theta \big) (\omega) := \left\{ x + \theta(x) 
\mbox{ for } \ x \in \omega\right\} ,
$$
with $\theta\in W^{1,\infty}(\Omega;\RR^d)$ such that
$\left\|\theta\right\|_{W^{1,\infty}(\Omega;\RR^2)}<1$ and tangential on $\partial \Omega$ 
(i.e., $\theta\cdot n=0$ on $\partial \Omega$ ; this last
condition ensures that $\Omega=( Id + \theta)\Omega$).
It is well known that, for sufficiently small $\theta$,
$(Id + \theta)$ is a diffeomorphism in $\Omega$.

\begin{defn}
\label{defshape}
The shape derivative of a function $J(\omega)$ is
defined as the Fr\'echet derivative in $W^{1,\infty}(\Omega;\RR^d)$
at 0 of the application $\theta \to J\big( \big( Id + \theta \big) \omega \big)$,
i.e.
$$
J\big( \big( Id + \theta \big) \omega \big)  = J (\omega)
+ J^\prime(\omega)(\theta) + o(\theta) \quad\mbox{with}\quad
\lim_{\theta\to 0} \frac{|o(\theta)|}{\quad \Vert\theta\Vert_{W^{1,\infty}}}=0\;,
$$
where $J^\prime(\omega)$ is a continuous linear form on $W^{1,\infty}(\Omega;\RR^d)$.\\
\end{defn}

Now we apply the above definition in the context of our cell problem \eqref{cell_problem}.

{
\lem
Let $\omega$ be a smooth subset of $Y$. Define
\begin{equation*}
J_{surf}(\omega)=|\partial \omega|=\int_{\partial \omega} ds \;\;\;\mbox{   and  }\;\;\; J_{vol}(\omega)=|\omega|=\int_{\omega} dx.
\end{equation*}

Then the respective shape gradients at $0$ in the direction $\theta$ read
\begin{eqnarray*}
J'_{surf}(\omega)(\theta)=\int_{\partial \omega} \theta \cdot n H\,ds,\quad 
J'_{vol}(\omega)(\theta)=\int_{\partial \omega} \theta \cdot n\, ds,
\end{eqnarray*}
where $H$ is the mean curvature of $\partial \omega$ defined by $H=\mbox{div}(n)$.\\
}

{
\prop Define the spaces
\begin{eqnarray*}
&&H_{0,\#}^1(Y\backslash \omega)^d= \left\{ v\in H^1(Y\backslash \omega)^d: v|_{\partial \omega}=0, y\rightarrow v(y) \mbox{ is $Y-$periodic}\right\},\\
&&L_{0,\#}^2(Y\backslash \omega)= \left\{ q\in L^2(Y\backslash \omega): \int_{Y\backslash \omega} q dx=0,\; y\rightarrow q(y) \mbox{ is $Y-$periodic}\right\}.
\end{eqnarray*}
Let $u_i\in H_{0,\#}^1(Y\backslash \omega)^d$ and $p_i \in L^2_{0,\#}(Y\backslash \omega)$, $i=1,...,d$ be the collection of solutions of the variational formulation of the first cell problem in (\ref{cell_problem})
\begin{equation}
\label{variational_stokes}
\int_{Y\backslash \omega} \Big(\nabla u_i : \nabla v_i - \text{div}(v_i)p_i - \mbox{div}(u_i) q_i \Big)dx = \int_{Y\backslash \omega} e_i \cdot v_i dx, \;\;\forall (v_i,q_i)\in (H_{0,\#}^1(Y\backslash \omega)^d,L^2_{0,\#}(Y\backslash \omega)).
\end{equation}
Define the cost function (which does not depend on the pressure)
\begin{equation*}
J(\omega)=\int_{Y\backslash \omega} j(x,u,\nabla u)dx, \text{ where }u=(u_i)_{i=1...d},
\end{equation*}
and the collection of adjoint states $(U_i,P_i)\in (H_{0,\#}^1(Y\backslash \omega)^d,L^2_{0,\#}(Y\backslash \omega))$, such that  
\begin{equation}
\label{variational_adjoint}
\int_{Y\backslash \omega} \Big(\nabla U_i : \nabla v_i -\text{div}(v_i)P_i - \mbox{div}(U_i) q_i \Big)dx = -\int_{Y\backslash \omega} \Big( \frac{\partial j}{\partial u_i}(x,u_i,\nabla u_i) \cdot v_i+ \frac{\partial j}{\partial \nabla u_i}(x,u_i,\nabla u_i):\nabla v_i\Big) dx,
\end{equation}
{for all $(v_i,q_i)\in (H_{0,\#}^1(Y\backslash \omega)^2,L^2_{0,\#}(Y\backslash \omega))$.} Then $J(\omega)$ is shape differentiable at $0$ in the direction $\theta$
and the shape derivative reads 
\begin{equation*}
J'(\omega)(\theta)=\int_{\partial \omega} \Big\{j(x,u,\nabla u)-\sum_{i=1}^d \Big(\frac{\partial U_i}{\partial n}\cdot \frac{\partial u_i}{\partial n}+\frac{\partial j}{\partial \nabla u_i}\cdot n \cdot \frac{\partial u_i}{\partial n} \Big)\Big\}\theta \cdot n  ds.
\end{equation*}

}
{
\cor
According to \eqref{def:tensors}
\begin{equation*}
\frac{\mbox{tr(K)}}{d}=\frac{1}{d}\sum_{i=1}^d K_{ii}=\frac{1}{d}\int_{Y\backslash \omega}\sum_{i=1}^d |\nabla u_i|^2dx,
\end{equation*}
so the shape derivative reads
\begin{equation*}
\frac{1}{d}\text{tr}(K)'(\omega)(\theta)=-\frac{1}{d}\sum_{i=1}^d\int_{\partial \omega}  \Big(\big(\frac{\partial u_i}{\partial n}\big)^2+\frac{\partial  u_i}{\partial n} \cdot \frac{\partial U_i}{\partial n}\Big)\theta \cdot nds,
\end{equation*}
 

with $U_i$ solution of the adjoint problem 
\begin{equation*}
\int_{Y\backslash \omega}\Big(\nabla U_i : \nabla v_i -\text{div}(v_i)P_i - \mbox{div}(U_i) q_i\Big) dx =-2 \int_{Y\backslash \omega} \nabla u_i: \nabla v_i,
\forall (v_i,q_i)\in (H_{0,\#}^1(Y\backslash \omega)^2,L^2_{0,\#}(Y\backslash \omega))
\end{equation*}
}

\section{Numerical results}
\label{numerical_results_fuel_cell}

Now we detail the numerical solution of problem (\ref{optim_problem}) {via the} level-set method of Osher and Sethian \cite{osher.sethian} for topology optimization. At each iteration of the algorithm, the shape of the tube $\omega$ in the unit cell $Y$ is parametrized via a level-set function
$$
\left\{ \begin{array}{ll}
\psi(x) = 0 & \mbox{ for } x\in\partial \omega,  \\
\psi(x) < 0 &  \mbox{ for } x\in\omega, \\
\psi(x) > 0 & \mbox{ for } x\in Y\backslash\omega,
\end{array}
\right.
$$
defined on a fixed mesh in an Eulerian framework. 

During the optimization process, the shape $\omega(t)$ is going to evolve according to a fictitious time $t\in \RR^+$, which corresponds to the descent stepping. As it is well known, the evolution of the level set function $\psi$ is governed by a Hamilton-Jacobi equation 
\begin{equation}
\label{HJ_eq}
\frac{\partial \psi}{\partial t}+\mathcal{V} |\nabla \psi|=0,
\end{equation}
where $\mathcal{V}$ is the normal velocity. Equation (\ref{HJ_eq}) is posed in the whole domain $Y$, and not only on the interface $\partial \omega$. 
The main point in using a level set method is that it replaces the Lagrangian evolution of the boundary $\partial \omega$ by the Eulerian solution of a transport equation in the whole fixed
domain $Y$, easily allowing topology changes. In practice the Hamilton-Jacobi equation (\ref{HJ_eq}) is solved by an explicit second order upwind scheme (see e.g. \cite{sethian.99}) on a fixed Cartesian grid. The boundary conditions for $\psi$ are of Neumann type. Since this scheme is explicit in time, its time stepping must satisfy a CFL condition. Moreover, in order to regularize the
level set function (which may become too flat or too steep), we reinitialize it periodically by solving another Hamilton-Jacobi equation which admits as a stationary solution the signed distance to the initial interface \cite{sethian.99}.

\medskip

The choice of the normal velocity $\mathcal{V}$ in (\ref{HJ_eq}) 
is based on the shape derivatives computed in the previous section accordingly to  \cite{allaire2004structural}.
Generally speaking, the shape derivative of a functional $J(\omega)$ in 
the direction $\theta\in W^{1,\infty}(Y;\RR^d)$ reads
\begin{equation}
\label{general_shape_derivative}
J'(\omega)(\theta)=\int_{\partial \omega}\mathcal{T}\theta \cdot n ds,
\end{equation}
where the integrand $\mathcal{T}(x)$ depends on the solution of the cell problem and possibly of some adjoint equation. 
Since only the normal component of $\theta$ plays a role in (\ref{general_shape_derivative}), 
a descent direction for $J$ is just a vector field $\theta$ 
satisfying
\begin{equation}
\theta = \mathcal{V} n \quad \mbox{ and } \quad 
J'(\omega)(\theta)=\int_{\partial \omega}\mathcal{T} \mathcal{V} ds\leq 0.
\end{equation}
To ensure the decrease of $J$, the simplest choice is $\mathcal{V}=-\mathcal{T}$. 
However, $\mathcal{T}$ is a priori defined only on the interfaces $\partial \omega$ 
while $\mathcal{V}$ must be defined in the entire domain $Y$. Therefore, the 
choice $\mathcal{V}=-\mathcal{T}$ is implicitly depending on some extension 
process of $\mathcal{T}$. If such an extension is not obvious or if 
we want to regularize the velocity fields, there is an alternate choice 
based on a different underlying scalar product (see e.g. 
\cite{allaire2004structural}, \cite{MR2225309}). \\

Since the interface $\partial \omega$ does not necessarily match with the mesh, in order to simulate an impenetrable interface in \eqref{cell_problem}, an artificial permeability (or porosity) term is added on the {left hand sides} of equations \eqref{variational_stokes} and \eqref{variational_adjoint}, respectively 
\begin{equation}
\label{artificial_momentum}
\int_{Y}\frac{5}{2\rho^2}(u_i \cdot v_i)dx,\quad \int_{Y}\frac{5}{2\rho^2}(U_i \cdot v_i)dx,
\end{equation}
where $2\rho^2/5$ represents the permeability of the medium \cite{brinkman1947calculation,durlofsky1987analysis}. For topology optimization purposes, $\rho$ is a regular approximation of $\mathcal{H}(\psi)$, where $\mathcal{H}$ is the Heaviside function and $\min \limits_{x\in Y} \rho (x)=\delta>0$ \cite{bendsoe1995methods,borvall_petersson_2002}. Typically this value is set up to $\delta=0.001$. We remark that thanks to the introduction of  \eqref{artificial_momentum}, all integrals are thus calculated in $Y$.

\begin{figure}[h]
\centering
\includegraphics[width=0.32\textwidth,height=0.32\textwidth]{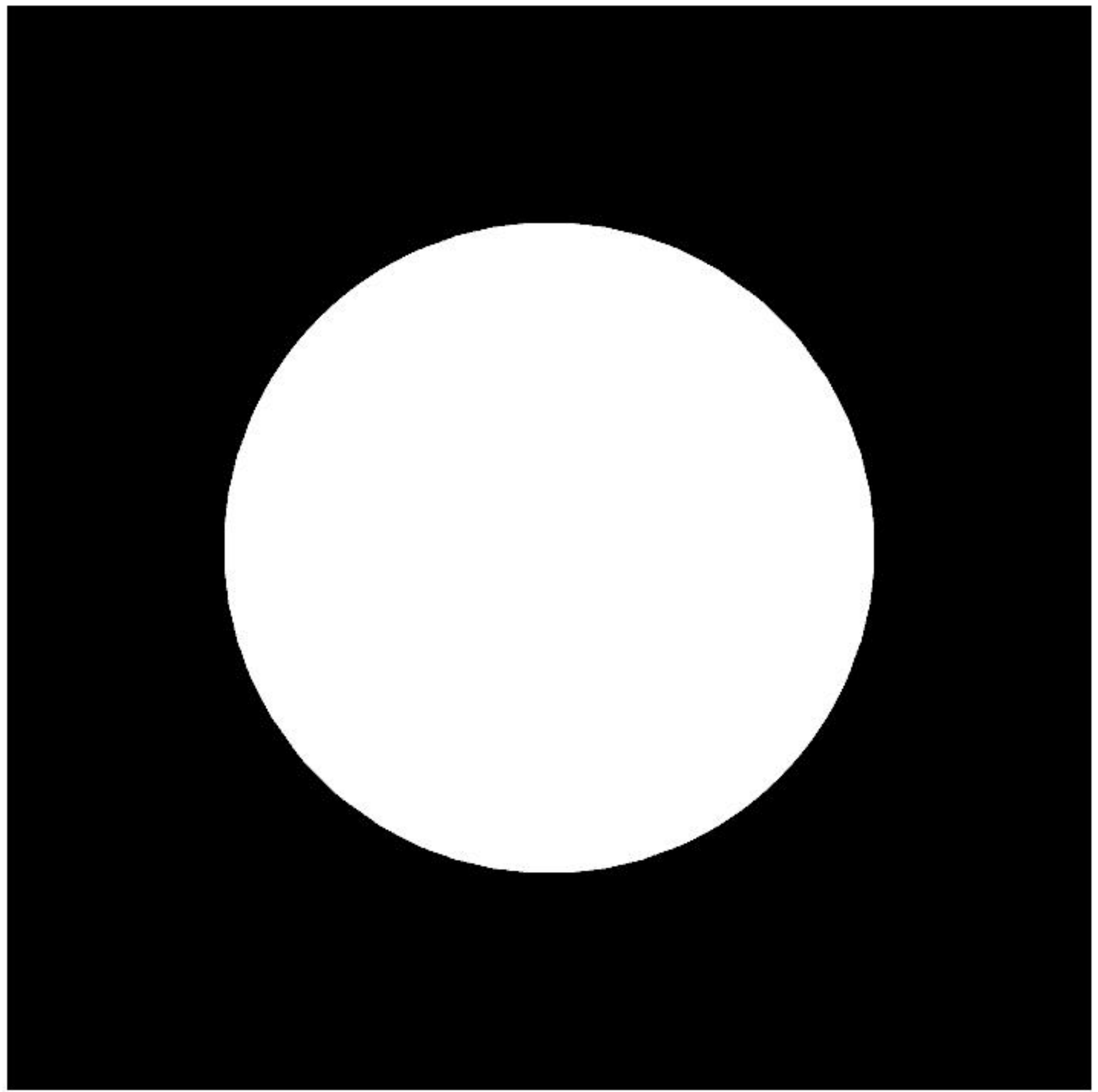}
\includegraphics[width=0.32\textwidth,height=0.32\textwidth]{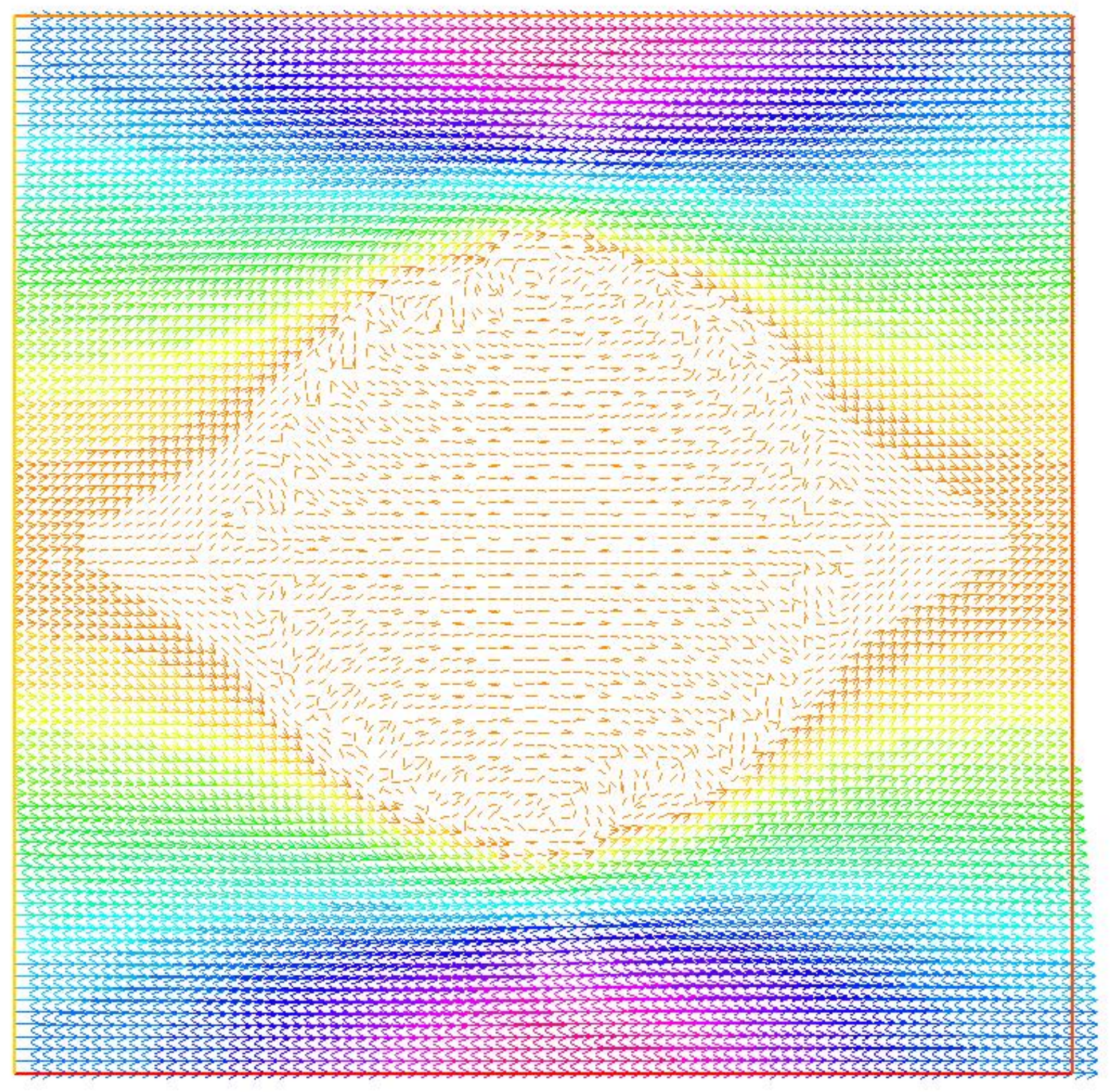}
\includegraphics[width=0.32\textwidth,height=0.32\textwidth]{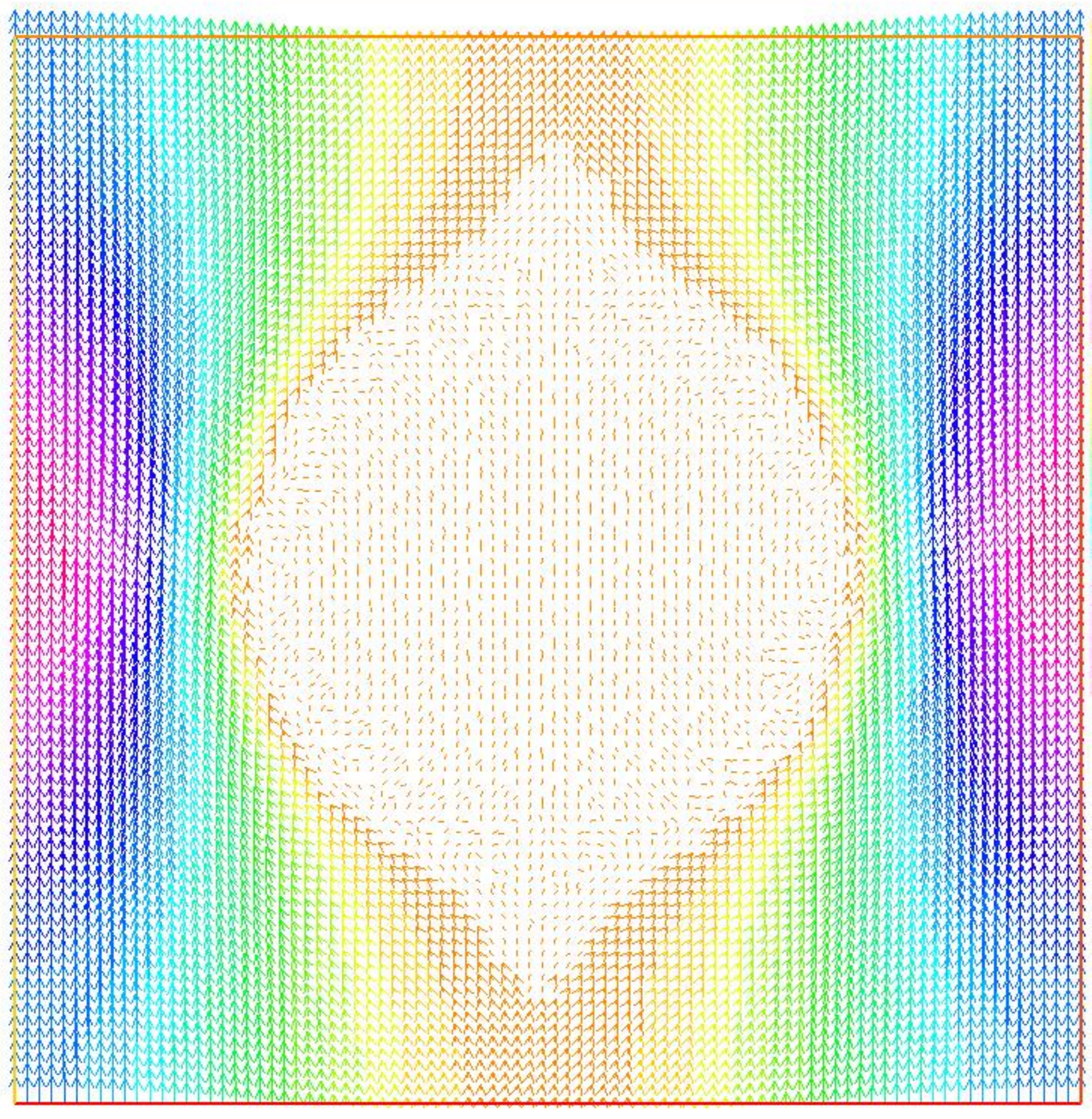}
\caption{Example of the 2D stokes solutions $u_1,u_2$ (left and right respectively) of the cell problem for a circular domain in \eqref{variational_stokes}. The black zone represents the air phase (transporting the oxygen) meanwhile the white one the fuel phase. A solid interface lays between them. We remark that inside the fuel phase the velocity of the air is almost zero, due to \eqref{artificial_momentum}.}
\end{figure}

For the optimization algorithm we use an augmented Lagrangian method $$
\mathcal{L}(\omega,\ell,\mu) = |\partial \omega|-\ell_{1}(C_f|\partial \omega|-|\omega|)-\ell_2 \left (k_{min}-\frac{\text{tr}(K)}{d} \right)+
\frac{\mu_1}{2}(C_f|\partial \omega|-|\omega|)^2+\frac{\mu_2}{2} \left(k_{min}-\frac{\text{tr}(K)}{d} \right)^2,
$$  
where $\ell=(\ell_i)_{i=1,2}$ and $\mu=(\mu_i)_{i=1,2}$ are lagrange multipliers and penalty parameters for the constraints. The Lagrange multipliers are updated at each iteration $n$ according to the optimality condition
\[
\ell_{1}^{n+1}=\ell_{1}^{n}-\mu_{1}(C_f|\partial \omega_n|-|\omega_n|)
\quad \mbox{and} \quad \ell_{2}^{n+1}=\ell_{2}^{n}-\mu_{2}\left(k_{min}-\frac{\text{tr}(K)}{d}\right).
\]
 The penalty parameters are augmented every $5$ iterations. With such an algorithm the constraints are enforced only at convergence.\\ 
 
Fig. \ref{homogenized} shows different optimal periodic layouts starting from an intuitive circular arrangement. The radius of the circle in $Y$ was set up to $r=0.3$, meanwhile the coefficients $C_f=0.15$ and $k_{min}=0.011$ were calculated so as both constraints in \eqref{optim_problem} were active within this layout. The first optimal design (second row Fig. \ref{homogenized}) corresponds to the optimal solution of the algorithm for the foregoing parameters. Then the second design (third row Fig. \ref{homogenized}) was established by reducing $C_f$ and $k_{min}$ to the half. The perimeter gain w.r.t. the circular layout is $7\%$ for the former design and $60\%$ for the latter one.\\

%

\begin{figure}[h]
\centering
\begin{tikzpicture}[      
        every node/.style={anchor=south west,inner sep=0pt},
        x=1mm, y=1mm,
      ]   
     \node (fig5) at (0cm,0cm){\includegraphics[width=0.25\textwidth,height=0.24\textwidth]{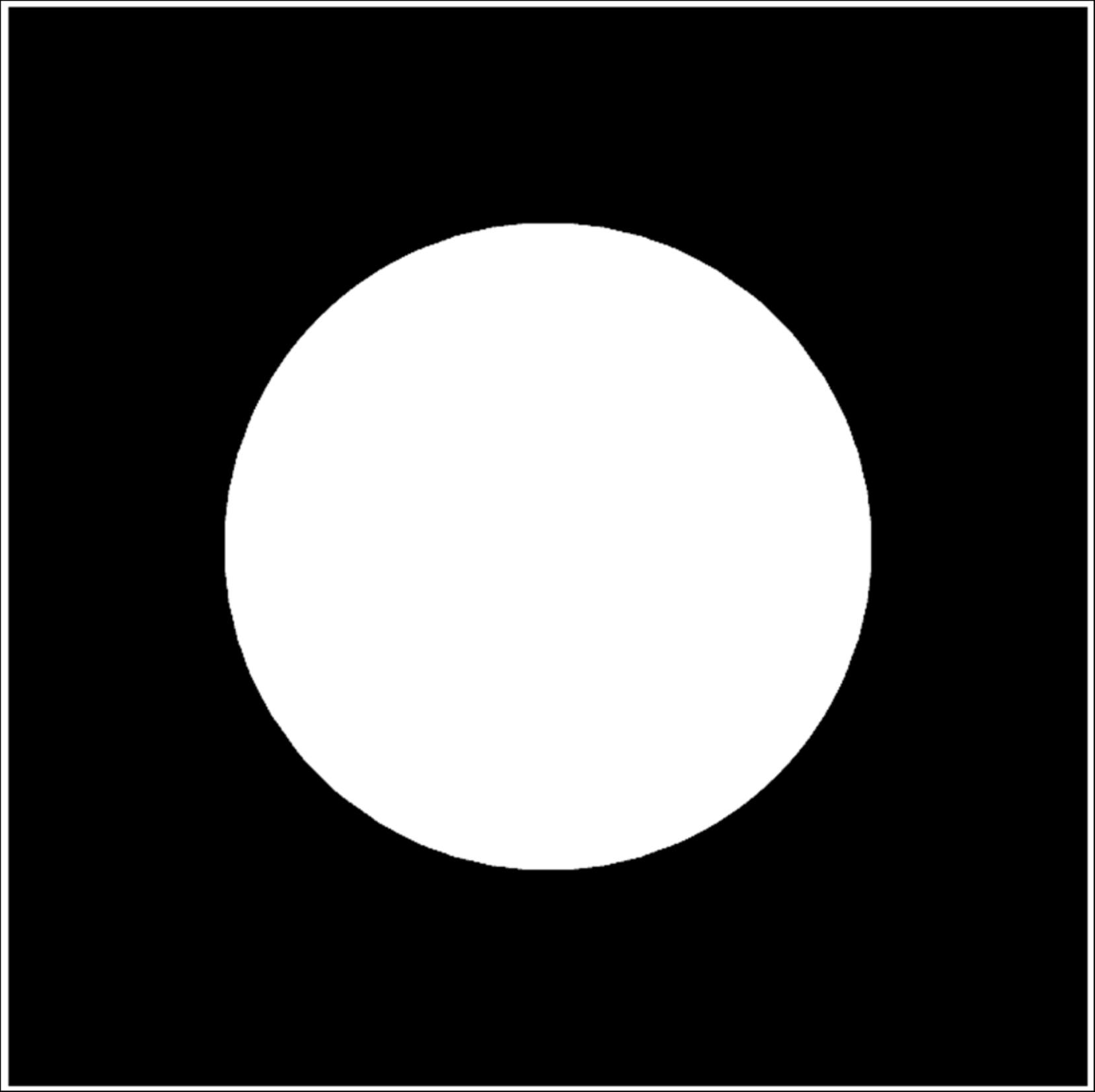}};
     \node (fig6) at (4.5cm,0cm)
     {\includegraphics[width=0.25\textwidth,height=0.24\textwidth]{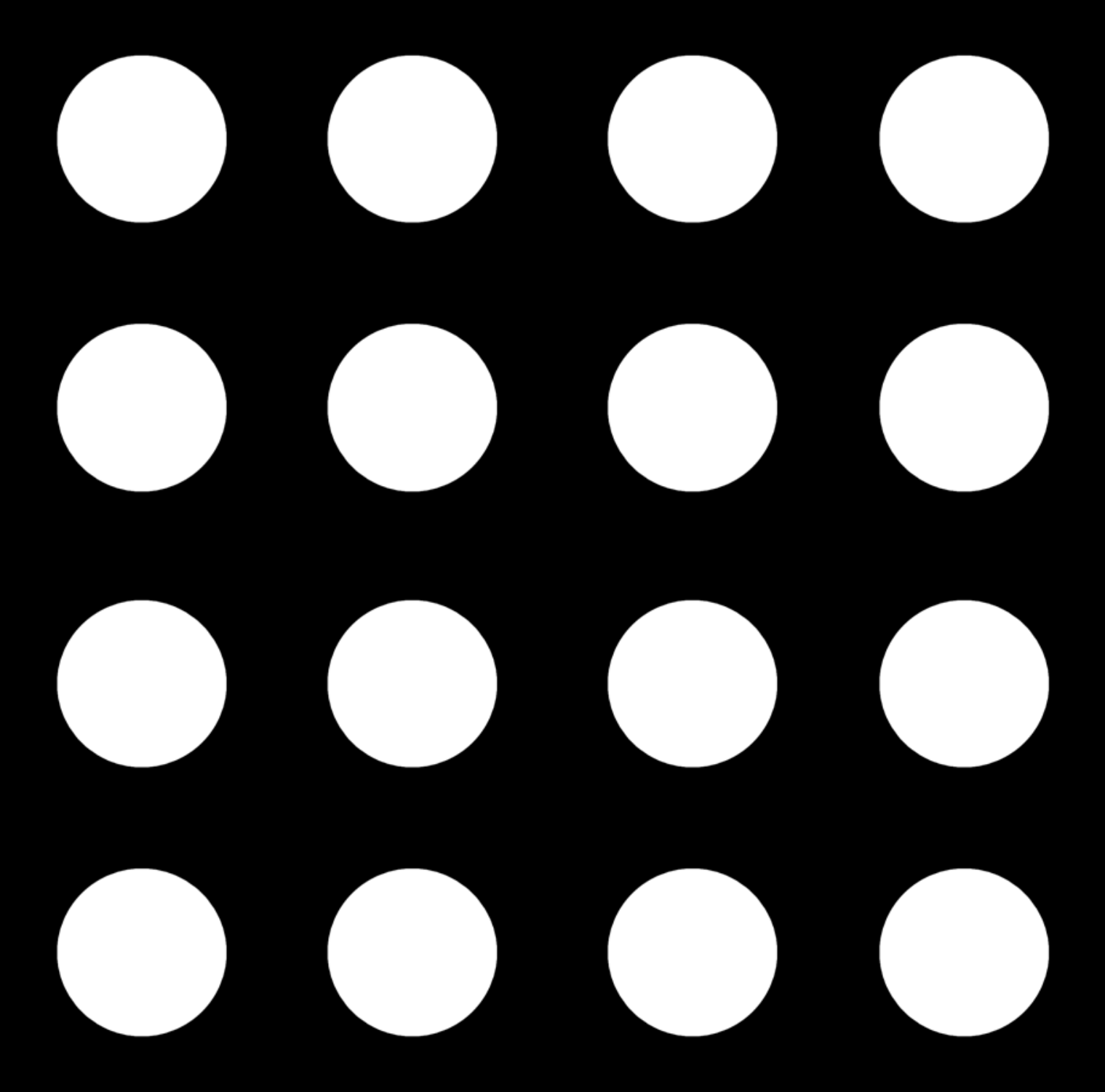}}; 
     \node (fig1) at (0cm,-4.5cm){\includegraphics[width=0.25\textwidth,height=0.24\textwidth]{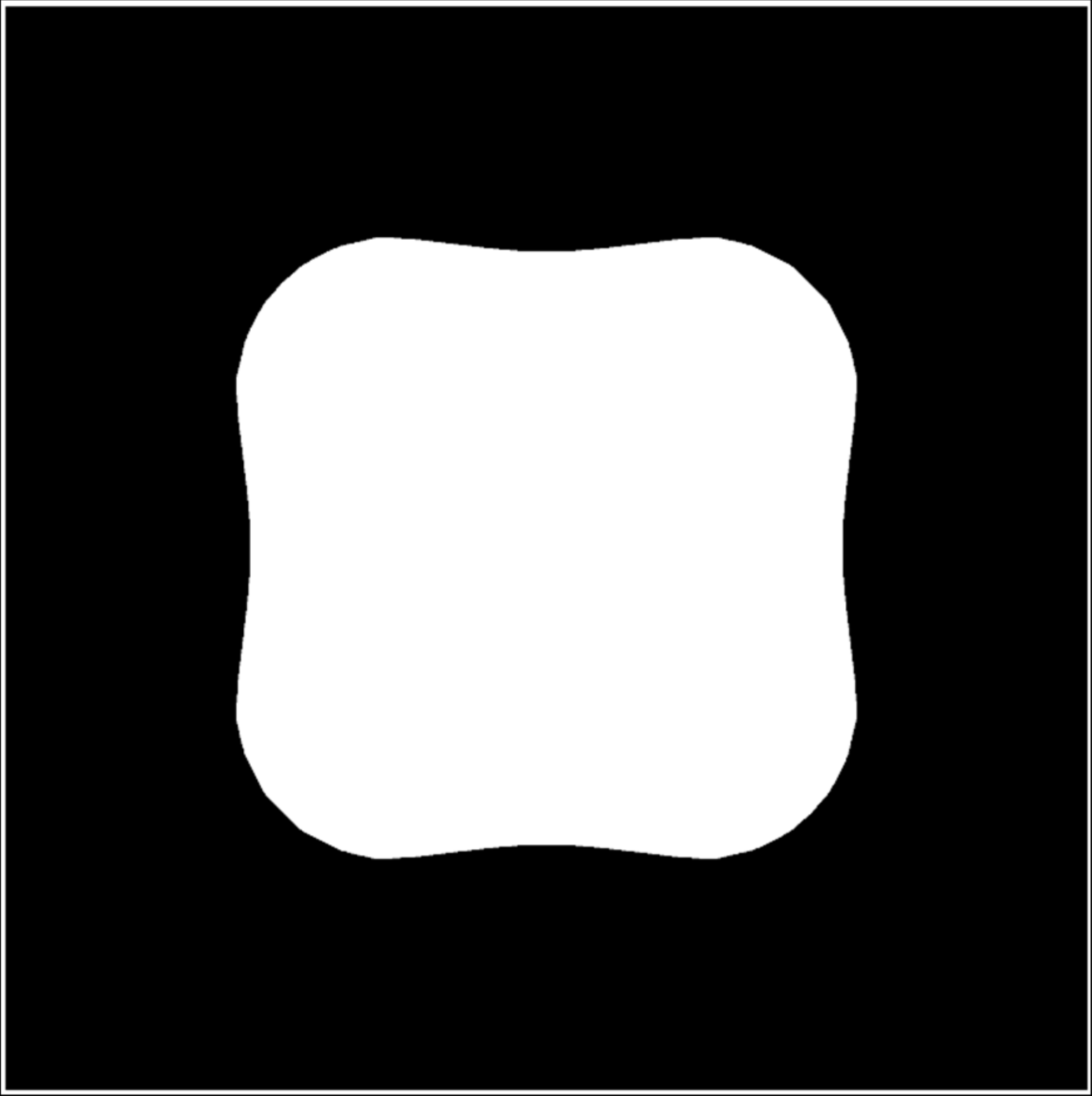}};
     \node (fig2) at (4.5cm,-4.5cm)
     {\includegraphics[width=0.25\textwidth,height=0.24\textwidth]{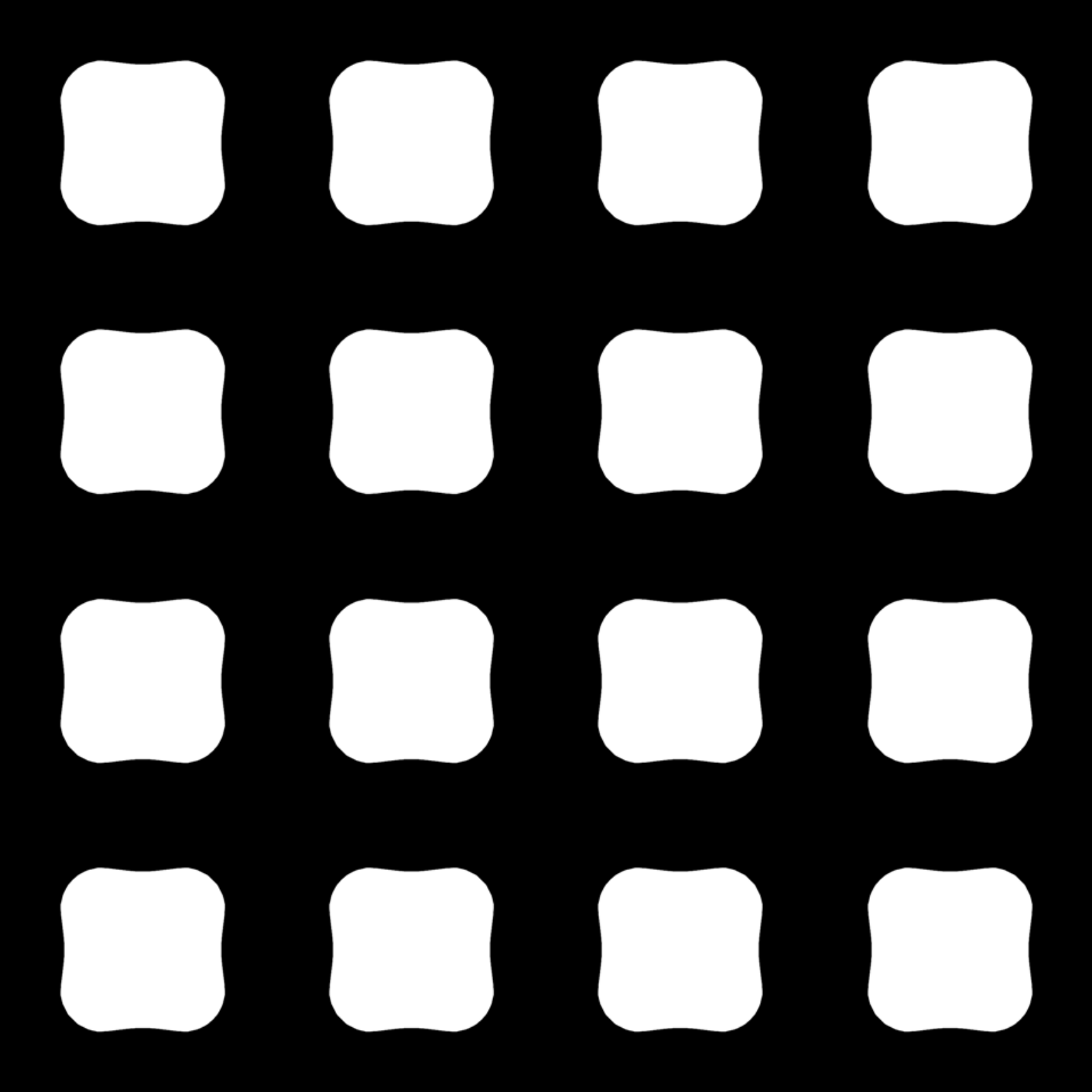}};
     \node (fig3) at (0cm,-9cm){\includegraphics[width=0.25\textwidth,height=0.24\textwidth]{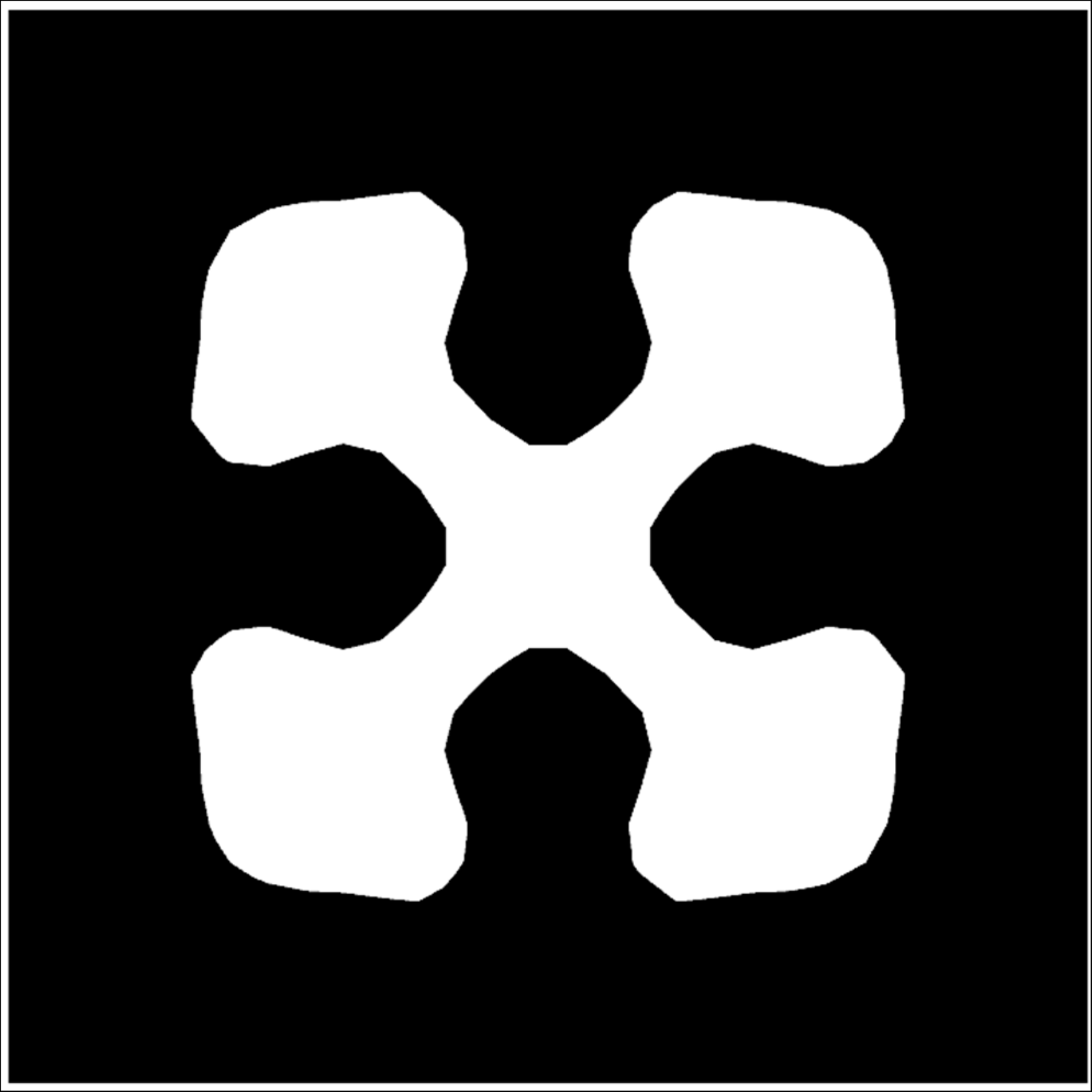}};
     \node (fig4) at (4.5cm,-9cm)
     {\includegraphics[width=0.25\textwidth,height=0.24\textwidth]{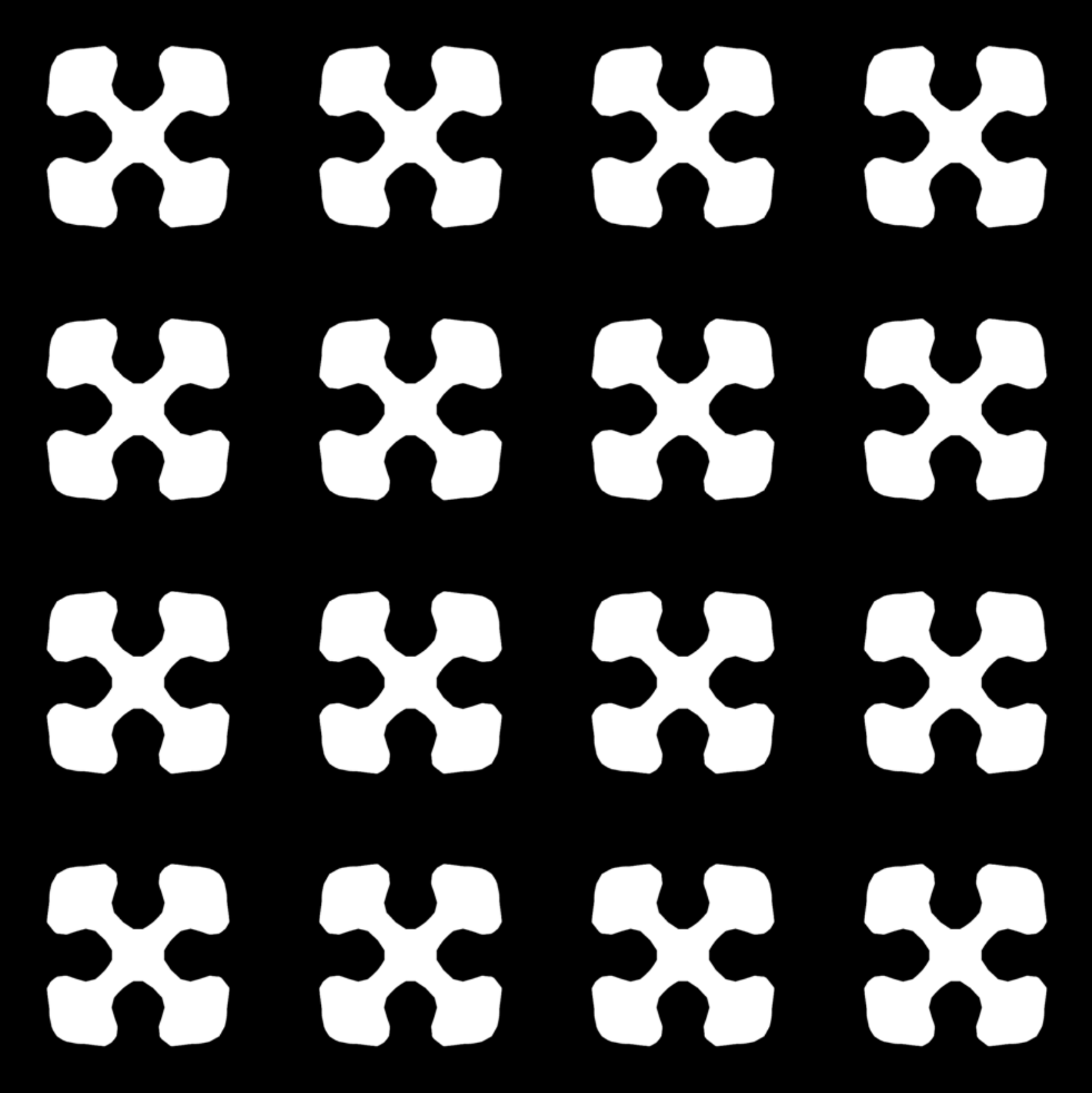}};
\end{tikzpicture}
\caption{Two micro-tubular optimal designs for different values of $k_{min}$ and $C_f$ starting from a circular layout. The base cell results are on the left and the corresponding periodic structures on the right.}\label{homogenized}
\end{figure}

{\section{Summary and Outlook}}

The above results demonstrate the capabilities of topology optimization via a level-set method to enhance the design of micro-tubular fuel cells for the aeronautic industry. Thanks the application of an inverse homogenization technique and the level-set method, periodic optimal micro-tubular fuel cells with a sharp contour can easily be designed and then manufactured by 3D printing. The foregoing study thus suggests a promising use of these technologies in the future  computer aided design of fuel cells.\\

\textbf{Acknowledgments:} {The author would like to thank Airbus Group for its financial support in the framework of his PhD thesis, Ch. Nespoulous and E. Moullet of Airbus Innovations for their important contribution through the joint work developed during E. Mullet's internship \cite{etienne}, and also G. Allaire of Centre de Math\'ematiques Appliqu\'ees of \'Ecole Polytechnique for many helpful suggestions after kindly reading this manuscript.}

\newpage
\bibliographystyle{abbrv}
\bibliography{FC_biblio,background}

\end{document}

%% file: macros.tex
\newtheorem{thm}{Theorem}[section]

\newtheorem{defn}[thm]{Definition}
\newtheorem{prop}[thm]{Proposition}
\newtheorem{lem}[thm]{Lemma}
\newtheorem{cor}[thm]{Corollary}

\newcommand{\ue}{u_\epsilon}
\newcommand{\pe}{p_\epsilon}

\def\RR{{\mathbb{R}}}
\def\div{ {\rm div}}

\def\Uadmis{{\mathcal U}_{ad}}

\newcommand{\Rbb} {\mathbb{R}}






%% file: SMO.bbl
\begin{thebibliography}{10}

\bibitem{allaire2002shape}
G.~Allaire.
\newblock {\em Shape optimization by the homogenization method}, volume 146.
\newblock Springer, 2002.

\bibitem{allaire2006conception}
G.~Allaire.
\newblock {\em Conception optimale de structures}, volume~58.
\newblock Springer Verlag, 2006.

\bibitem{allaire1997shape}
G.~Allaire, E.~Bonnetier, G.~Francfort, and F.~Jouve.
\newblock Shape optimization by the homogenization method.
\newblock {\em Numerische Mathematik}, 76(1):27--68, 1997.

\bibitem{allaire2004structural}
G.~Allaire, F.~Jouve, and A.-M. Toader.
\newblock Structural optimization using sensitivity analysis and a level-set
  method.
\newblock {\em Journal of computational physics}, 194(1):363--393, 2004.

\bibitem{allaire1993optimal}
G.~Allaire and R.~V. Kohn.
\newblock Optimal design for minimum weight and compliance in plane stress
  using extremal microstructures.
\newblock {\em European journal of mechanics. A. Solids}, 12(6):839--878, 1993.

\bibitem{auriault2010homogenization}
J.-L. Auriault, C.~Boutin, and C.~Geindreau.
\newblock {\em Homogenization of coupled phenomena in heterogenous media},
  volume 149.
\newblock John Wiley \& Sons, 2010.

\bibitem{bendsoe1995methods}
M.~Bends{\o}e.
\newblock {\em Methods for optimization of structural topology, shape and
  material}.
\newblock Springer Verlag, New York, 1995.

\bibitem{bendsoe2004topology}
M.~Bends{\o}e and O.~Sigmund.
\newblock {\em Topology optimization: theory, methods and applications}.
\newblock Springer, 2004.

\bibitem{bendsoe1988generating}
M.~P. Bends{\o}e and N.~Kikuchi.
\newblock Generating optimal topologies in structural design using a
  homogenization method.
\newblock {\em Computer methods in applied mechanics and engineering},
  71(2):197--224, 1988.

\bibitem{borvall_petersson_2002}
T.~Borrvall and J.~Petersson.
\newblock Topology optimization of fluids in stokes flow.
\newblock {\em International Journal for Numerical Methods in Fluids},
  41(1):77--107, 2003.

\bibitem{bove2008modeling}
R.~Bove and S.~Ubertini.
\newblock {\em Modeling solid oxide fuel cells: methods, procedures and
  techniques}, volume~1.
\newblock Springer Science+ Business Media, 2008.

\bibitem{brinkman1947calculation}
H.~Brinkman.
\newblock A calculation of the viscosity and the sedimentation constant for
  solutions of large chain molecules taking into account the hampered flow of
  the solvent through these molecules.
\newblock {\em Physica}, 13(8):447--448, 1947.

\bibitem{burger2003framework}
M.~Burger.
\newblock A framework for the construction of level set methods for shape
  optimization and reconstruction.
\newblock {\em Interfaces and Free boundaries}, 5(3):301--330, 2003.

\bibitem{cadman2013design}
J.~E. Cadman, S.~Zhou, Y.~Chen, and Q.~Li.
\newblock On design of multi-functional microstructural materials.
\newblock {\em Journal of Materials Science}, 48(1):51--66, 2013.

\bibitem{cherkaev2000variational}
A.~Cherkaev.
\newblock {\em Variational methods for structural optimization}, volume 140.
\newblock Springer, 2000.

\bibitem{costamagna2004electrochemical}
P.~Costamagna, A.~Selimovic, M.~Del~Borghi, and G.~Agnew.
\newblock Electrochemical model of the integrated planar solid oxide fuel cell
  (ip-sofc).
\newblock {\em Chemical Engineering Journal}, 102(1):61--69, 2004.

\bibitem{danilov2009cfd}
V.~A. Danilov and M.~O. Tade.
\newblock A cfd-based model of a planar sofc for anode flow field design.
\newblock {\em International Journal of Hydrogen Energy}, 34(21):8998--9006,
  2009.

\bibitem{MR2225309}
F.~de~Gournay.
\newblock Velocity extension for the level-set method and multiple eigenvalues
  in shape optimization.
\newblock {\em SIAM J. Control Optim.}, 45(1):343--367, 2006.

\bibitem{de2011production}
R.~De~la Torre~Garc{\'\i}a.
\newblock {\em Production of Micro-Tubular Solid Oxide Fuel Cells}.
\newblock PhD thesis, University of Trento,
  http://eprints-phd.biblio.unitn.it/541/1/PRODUCTION\textunderscore OF
  \textunderscore MICRO-TUBULAR \textunderscore SOLID \textunderscore OXIDE
  \textunderscore FUEL \textunderscore CELLS.pdf, 2011.

\bibitem{durlofsky1987analysis}
L.~Durlofsky and J.~Brady.
\newblock Analysis of the brinkman equation as a model for flow in porous
  media.
\newblock {\em Physics of Fluids}, 30(11):3329--3341, 1987.

\bibitem{gou2009fuel}
B.~Gou, W.~K. Na, and B.~Diong.
\newblock {\em Fuel cells: modeling, control, and applications}.
\newblock CRC press, 2009.

\bibitem{guest2007design}
J.~K. Guest and J.~H. Pr{\'e}vost.
\newblock Design of maximum permeability material structures.
\newblock {\em Computer Methods in Applied Mechanics and Engineering},
  196(4):1006--1017, 2007.

\bibitem{haslinger1995optimum}
J.~Haslinger and J.~Dvorak.
\newblock Optimum composite material design.
\newblock {\em RAIRO-M2AN Modelisation Math et Analyse Numerique-Mathem Modell
  Numerical Analysis}, 29(6):657--686, 1995.

\bibitem{pierre}
A.~Henrot and M.~Pierre.
\newblock {\em Variation et Optimisation des formes, Une analyse
  g{\'e}om{\'e}trique}, volume~48 of {\em Math{\'e}matiques et Applications}.
\newblock Springer, 2005.

\bibitem{hordeski2009hydrogen}
M.~F. Hordeski.
\newblock {\em Hydrogen and Fuel Cells: Advances in Transportation and Power}.
\newblock The Fairmont Press, Inc., 2009.

\bibitem{hornung1997homogenization}
U.~Hornung.
\newblock {\em Homogenization and porous media}, volume~6.
\newblock Springer, 1997.

\bibitem{hornung1991diffusion}
U.~Hornung and W.~J{\"a}ger.
\newblock Diffusion, convection, adsorption, and reaction of chemicals in
  porous media.
\newblock {\em Journal of differential equations}, 92(2):199--225, 1991.

\bibitem{hosseini2013cfd}
S.~Hosseini, K.~Ahmed, and M.~O. Tad{\'e}.
\newblock Cfd model of a methane fuelled single cell sofc stack for analysing
  the combined effects of macro/micro structural parameters.
\newblock {\em Journal of Power Sources}, 2013.

\bibitem{hussain2006mathematical}
M.~Hussain, X.~Li, and I.~Dincer.
\newblock Mathematical modeling of planar solid oxide fuel cells.
\newblock {\em Journal of Power Sources}, 161(2):1012--1022, 2006.

\bibitem{Kee2012331}
R.~J. Kee, H.~Zhu, R.~J. Braun, and T.~L. Vincent.
\newblock Chapter 6 - modeling the steady-state and dynamic characteristics of
  solid-oxide fuel cells.
\newblock In K.~Sundmacher, editor, {\em Fuel Cell Engineering}, volume~41 of
  {\em Advances in Chemical Engineering}, pages 331 -- 381. Academic Press,
  2012.

\bibitem{kenney2004mathematical}
B.~Kenney and K.~Karan.
\newblock Mathematical micro-model of a solid oxide fuel cell composite
  cathode.
\newblock {\em Proceedings hydrogen and fuel cells}, pages 1--11, 2004.

\bibitem{kohn1986optimal}
R.~V. Kohn and G.~Strang.
\newblock Optimal design and relaxation of variational problems; {P}arts {I},
  {II}, {III}.
\newblock {\em Communications on Pure and Applied Mathematics}, 39(1):113--137,
  1986.

\bibitem{kumar2003modeling}
A.~Kumar and R.~Reddy.
\newblock Modeling of polymer electrolyte membrane fuel cell with metal foam in
  the flow-field of the bipolar/end plates.
\newblock {\em Journal of power sources}, 114(1):54--62, 2003.

\bibitem{etienne}
E.~Moullet.
\newblock Optimisation du design d'une pile \`a combustible.
\newblock Technical report, \'Ecole de Mines, EADS, 2013.

\bibitem{murat1974quelques}
F.~Murat and J.~Simon.
\newblock {\em Quelques r{\'e}sultats sur le contr{\^o}le par un domaine
  g{\'e}om{\'e}trique}.
\newblock VI Laboratoire d'Analyse Num{\'e}rique, 1974.

\bibitem{national_institute}
{National Institute of Advanced Industrial Science and Technology}.
\newblock https://www.aist:go:jp/index, 2009.

\bibitem{osher.sethian}
S.~Osher and J.~Sethian.
\newblock Fronts propagating with curvature-dependent speed: algorithms based
  on hamilton-jacobi formulations.
\newblock {\em Journal of computational physics}, 79(1):12--49, 1988.

\bibitem{palsson2000combined}
J.~Palsson, A.~Selimovic, and L.~Sjunnesson.
\newblock Combined solid oxide fuel cell and gas turbine systems for efficient
  power and heat generation.
\newblock {\em Journal of Power Sources}, 86(1):442--448, 2000.

\bibitem{roth2010fuel}
B.~Roth and R.~Giffin~III.
\newblock Fuel cell hybrid propulsion challenges and opportunities for
  commercial aviation.
\newblock {\em AIAA Paper}, 6537:2010, 2010.

\bibitem{samuelsen2004fuel}
S.~Samuelsen.
\newblock Fuel cell/gas turbine hybrid systems.
\newblock {\em ASME International Gas Turbine Institute}, 2004.

\bibitem{schmidt2010shape}
S.~Schmidt and V.~Schulz.
\newblock Shape derivatives for general objective functions and the
  incompressible navier-stokes equations.
\newblock {\em Control and Cybernetics}, 39(3):677--713, 2010.

\bibitem{schmuck2013homogenization}
M.~Schmuck and P.~Berg.
\newblock Homogenization of a catalyst layer model for periodically distributed
  pore geometries in pem fuel cells.
\newblock {\em Applied Mathematics Research eXpress}, 2013(1):57--78, 2013.

\bibitem{sethian.99}
J.~Sethian.
\newblock {\em Level set methods and fast marching methods: evolving interfaces
  in computational geometry, fluid mechanics, computer vision, and materials
  science}, volume~3.
\newblock Cambridge university press, 1999.

\bibitem{sethian2000structural}
J.~A. Sethian and A.~Wiegmann.
\newblock Structural boundary design via level set and immersed interface
  methods.
\newblock {\em Journal of computational physics}, 163(2):489--528, 2000.

\bibitem{sigmund1994design}
O.~Sigmund.
\newblock {\em Design of material structures using topology optimization}.
\newblock PhD thesis, Technical University of Denmark Denmark, 1994.

\bibitem{sigmund1994materials}
O.~Sigmund.
\newblock Materials with prescribed constitutive parameters: an inverse
  homogenization problem.
\newblock {\em International Journal of Solids and Structures},
  31(17):2313--2329, 1994.

\bibitem{sigmund1997design}
O.~Sigmund and S.~Torquato.
\newblock Design of materials with extreme thermal expansion using a
  three-phase topology optimization method.
\newblock {\em Journal of the Mechanics and Physics of Solids},
  45(6):1037--1067, 1997.

\bibitem{singhal2000advances}
S.~Singhal.
\newblock Advances in solid oxide fuel cell technology.
\newblock {\em Solid state ionics}, 135(1):305--313, 2000.

\bibitem{intro_sokolowski}
J.~Sokolowski and J.-P. Zol\'esio.
\newblock {\em Introduction to shape optimization}.
\newblock Springer, 1992.

\bibitem{steffen2005solid}
C.~J. Steffen~Jr, J.~E. Freeh, and L.~M. Larosiliere.
\newblock Solid oxide fuel cell/gas turbine hybrid cycle technology for
  auxiliary aerospace power.
\newblock {\em Proceedings of the ASME Turbo Expo, Reno-Tahoe, United States},
  2005.

\bibitem{tartar2009general}
L.~Tartar.
\newblock {\em The general theory of homogenization: a personalized
  introduction}, volume~7.
\newblock Springer, 2009.

\bibitem{wang2003level}
M.~Y. Wang, X.~Wang, and D.~Guo.
\newblock A level set method for structural topology optimization.
\newblock {\em Computer methods in applied mechanics and engineering},
  192(1):227--246, 2003.

\bibitem{zhou2008computational}
S.~Zhou and Q.~Li.
\newblock Computational design of multi-phase microstructural materials for
  extremal conductivity.
\newblock {\em Computational Materials Science}, 43(3):549--564, 2008.

\end{thebibliography}
